\documentclass[10pt,journal,compsoc]{IEEEtran}
\usepackage{cite}
\usepackage{graphicx}
\usepackage[cmex10]{amsmath}
\usepackage{amssymb}
\usepackage{amsthm}
\usepackage{array}
\usepackage[hyphens]{url}
\usepackage{algorithm,algorithmic}
\usepackage{chemarrow}
\usepackage{multirow}
\usepackage{extarrows}
\usepackage{setspace}
\usepackage{booktabs}
\usepackage{tabulary}
\usepackage[shortlabels]{enumitem}
\usepackage{balance}
\usepackage{subfigure}
\usepackage[cmex10]{amsmath}
\usepackage[shortlabels]{enumitem}
\usepackage{amssymb}
\usepackage{amsthm}
\usepackage{multirow}
\usepackage{cite}
\usepackage{color}
\usepackage{setspace}
\usepackage{ragged2e}
\usepackage{stmaryrd} %\llbracket
\usepackage{tabularx}
\usepackage{makecell}
\usepackage{footnote}
\usepackage{threeparttable}
\usepackage{eqparbox}
\usepackage{bm}
\usepackage{soul}
\makeatletter
\newcommand{\ALOOP}[1]{\ALC@it\algorithmicloop\ #1
\begin{ALC@loop}}
\newcommand{\ENDALOOP}{\end{ALC@loop}\ALC@it\algorithmicendloop}

\newtheorem{theorem}{\textbf{\emph{Theorem}}}
\newtheorem{definition}{\textbf{\emph{Definition}}}

\newcommand{\main}{SecSkyline}
\newcommand{\csa}{$C_{\{1,2\}}$}

\newcommand{\revise}{\textcolor{black}}

\begin{document}
	
\title{SecSkyline: Fast Privacy-Preserving Skyline Queries over Encrypted Cloud Databases}
	
\author{Yifeng Zheng, Weibo Wang, Songlei Wang, Xiaohua Jia, \emph{Fellow, IEEE}, Hejiao Huang, \\ and Cong Wang, \emph{Fellow, IEEE}
		
		\IEEEcompsocitemizethanks{
			\IEEEcompsocthanksitem Yifeng Zheng, Weibo Wang, and Songlei Wang are with the School of Computer Science and Technology, Harbin Institute of Technology, Shenzhen, Guangdong 518055, China (e-mail: yifeng.zheng@hit.edu.cn, weibo.wang.hitsz@outlook.com,
			and songlei.wang@outlook.com).
			\IEEEcompsocthanksitem Xiaohua Jia is with the School of Computer Science and Technology, Harbin Institute of Technology, Shenzhen, Guangdong 518055, China, and also with the Department of Computer Science, City University of Hong Kong, Hong Kong, China (e-mail: csjia@cityu.edu.hk).
			
		    \IEEEcompsocthanksitem Hejiao Huang is with the School of Computer Science and Technology, Harbin Institute of Technology, Shenzhen, Guangdong 518055, China, and also with the Guangdong Provincial Key Laboratory of Novel Security Intelligence Technologies (e-mail: huanghejiao@hit.edu.cn).
		
			\IEEEcompsocthanksitem C. Wang is with the Department of Computer Science, City University of Hong Kong, Hong Kong, China (e-mail: congwang@cityu.edu.hk). 
			
			\IEEEcompsocthanksitem Corresponding author: Yifeng Zheng.
		}
	}
	\IEEEtitleabstractindextext{
		\begin{abstract}
		The well-known benefits of cloud computing have spurred the popularity of database service outsourcing, where one can resort to the cloud to conveniently store and query databases. Coming with such popular trend is the threat to data privacy, as the cloud gains access to the databases and queries which may contain sensitive information, like medical or financial data. A large body of work has been presented for querying encrypted databases, which has been mostly focused on secure keyword search. In this paper, we instead focus on the support for secure skyline query processing over encrypted outsourced databases, where little work has been done. Skyline query is an advanced kind of database query which is important for multi-criteria decision-making systems and applications. We propose SecSkyline, a new system framework building on lightweight cryptography for fast privacy-preserving skyline queries. SecSkyline ambitiously provides strong protection for not only the content confidentiality of the outsourced database, the query, and the result, but also for data patterns that may incur indirect data leakages, such as dominance relationships among data points and search access patterns. Extensive experiments demonstrate that SecSkyline is substantially superior to the state-of-the-art in query latency, with up to $813\times$ improvement.
        
		\end{abstract}
		
	}

	\maketitle

	\IEEEdisplaynontitleabstractindextext
	
	\IEEEpeerreviewmaketitle

\section{Introduction}

\label{sec:intro}

% Database query, as one of the most important technologies in the information technology field, has been widely studied in recent decades \cite{buhalis2008progress}.
%
Due to the well-known benefits of cloud computing \cite{QinW0018,JiangWHWLSR21}, there has been growing popularity of enterprises or organizations leveraging commercial clouds to store and query their databases (e.g., \cite{Cox,ADP,airbnb,PIXNET}, to list a few).
However, as databases may contain rich sensitive and proprietary information (like databases of medical records or financial records), deploying such database services in the cloud may raise critical privacy concerns.
Therefore, there is an urgent demand that security must be embedded in such database outsourcing services, providing protection for the information-rich databases, private queries, as well as query results.
In the literature, a large body of work has been presented for querying encrypted databases, which has been mostly focused on secure keyword search \cite{ghareh2018new,DautermanFLPS20,sun2021practical,gui2023rethinking}.

In this paper, we instead focus on secure skyline queries over outsourced databases, where little work has been done.
Given a query point, skyline query aims to retrieve a set of data points (called skyline points) which are not dominated by any other data point from a multi-dimensional database \cite{mohamed2008skyline}.
In particular, given a query point $\mathbf{q}$ and two data points $\mathbf{a}$ and $\mathbf{b}$ in the target database, $\mathbf{a}$ is said to dominate $\mathbf{b}$ if $\mathbf{a}$ is nearer to $\mathbf{q}$ than $\mathbf{b}$ at least in one dimension and not farther in other dimensions.
Skyline query is highly useful for multi-criteria decision-making systems in different domains, such as web information systems \cite{balke2004efficient}, wireless mobile ad-hoc networks \cite{huang2006skyline}, and geographical information systems \cite{deng2007multi}, especially when it is hard to define a single distance metric with all dimensions \cite{liu2017secure}.
%

% can benefit various applications such as customer information services, decision support and decision making systems \cite{borzsonyi2001skyline}\yifeng{ to be updated}.
%

\iffalse

As one of the most fundamental database query functionalities, skyline query aims to retrieve a small set of records (named as skyline tuples) which are not dominated by any other tuple from a large number of multi-dimensional data tuples \cite{mohamed2008skyline}.
%
It can benefit various applications such as customer information services, decision support and decision making systems \cite{borzsonyi2001skyline}.
%
In this paper, we focus on skyline query in the more advanced and practical setting, where a query tuple is issued and skyline tuples with respect to the query are expected to be returned \cite{Zou2010Dynamic}.
% the more advanced and practical skyline query--- dynamic skyline query, which is produce skyline tuples with respect to a query tuple \cite{Zou2010Dynamic}.
%
% Hereafter, we 
% 
In such setting, given a query tuple $\mathbf{q}$ and two tuples $\mathbf{a}$ and $\mathbf{b}$ in the target database, $\mathbf{a}$ is said to dominate $\mathbf{b}$ if $\mathbf{a}$ is nearer to $\mathbf{q}$ than $\mathbf{b}$ at least in one dimension and not farther in other dimensions.
\fi

To make the problem we focus on more concrete, we brief an example application to demonstrate how skyline query works.
Consider a medical institution who outsources its database of medical records to the cloud to share its diagnosis experiences.
Table \ref{table:OriData} shows the original database $\mathbf{P}$, where each record (i.e., a tuple) contains values for health-related attributes about a patient, including the respiratory rate (R) and heart rate (H). 
A doctor from another medical organization has a patient record with $(R=16, H=100)$, and wants to retrieve records from $\mathbf{P}$ whose conditions are similar to that of the patient based on skyline query processing.
% so as to find relevant treatment experiences.
%
Therefore, the doctor sends a query $\mathbf{q}=(16,100)$ to the cloud.
Upon receiving $\mathbf{q}$, the cloud first maps each record in $\mathbf{P}$ to $\mathbf{T}$ (i.e., Table \ref{table:MapData}) using the mapping function \cite{liu2017secure,liu2019secure} $\mathbf{t}_i[j]=|\mathbf{p}_i[j]-\mathbf{q}[j]|, i\in[1,4],j\in[1,2]$.
After that, the cloud finds intermediate skyline tuples $\{\mathbf{t}_{\star}\}$ from $\mathbf{T}$ if each $\mathbf{t}_{\star}$ cannot be dominated by any other tuple in $\mathbf{T}$.
The final returned patients' records (i.e., target skyline tuples) are $\mathbf{p}_{1}$ and $\mathbf{p}_{4}$, because $\mathbf{t}_{1}$ dominates $\mathbf{t}_{2}$ but does not dominate $\mathbf{t}_{3}$ and $\mathbf{t}_{4}$; and $\mathbf{t}_{4}$ dominates $\mathbf{t}_{3}$. Formal definition of dominance is given in Section \ref{subsec:Skyline}.
The dominance relationships in this example is illustrated in Fig. \ref{fig:Sample}.

The challenge that we aim to tackle in this paper is how to enable fast and privacy-preserving skyline queries over encrypted cloud databases.
With respect to the above application scenario as an example, we aim to allow the cloud hosting the database $\mathbf{P}$ in encrypted form to produce the skyline query result $\{\mathbf{p}_{1}, \mathbf{p}_{4}\}$ in encrypted form as well.
Meanwhile, besides ensuring the data content confidentiality for direct protection, it is also demanded that the cloud should be prevented from knowing data patterns which may cause indirect data leakage \cite{liu2017secure,liu2019secure,ding2021efficient}. 
Such data patterns include the dominance relationships among database tuples, the number of database tuples that each skyline tuple dominates, and the search access patterns.
Here the search pattern implies whether a new skyline query has been issued before and the access pattern reveals which database tuples are the skyline tuples.

In the literature, privacy-preserving skyline queries has recently received increasing attentions and several research endeavors have been proposed \cite{bothe2014skyline,liu2017secure,liu2019secure,ding2021efficient,wang2022efficient,wang2020stargazing,zhang2022efficient}.
The state-of-the-art prior works \cite{liu2019secure,ding2021efficient} that are most related to ours rely on heavy cryptosystems to craft secure protocols, leading to substantial performance overheads which heavily hinders the practical usability. 
In particular, even over very small-scale databases (e.g., with 1000 2-dimensional tuples), the state-of-the-art works \cite{liu2019secure,ding2021efficient} still require processing latency of more than \textit{1000} seconds and \textit{100} seconds, respectively.
Therefore, how to enable privacy-preserving skyline queries with practical performance is still challenging and remains to be fully explored.

In light of the above, in this paper, we propose {\main}, a new system framework that allows fast privacy-preserving skyline queries over encrypted databases outsourced to the cloud.
Different from prior arts \cite{liu2019secure,ding2021efficient}, {\main} fully utilizes a lightweight cryptographic technique---additive secret sharing \cite{demmler2015aby} and achieves substantially superior performance in query latency.
We first conduct an in-depth examination on the procedure of skyline query processing and identify that it can be decomposed into several essential components for which we provide customized secure realizations. 
%

% We start with an in-depth examination on the procedure of skyline query processing to identify phases which are computation-intensive and challenging to perform in ciphertext domain, and then devise a suite of secure and lightweight components to allow the cloud to obliviously perform skyline query with practical performance.

  \begin{table}[t!]
  \small
 	\begin{minipage}[t]{0.49\linewidth}
 		\centering
 		\caption{Original Database $\mathbf{P}$}
 		\label{table:OriData}
 		% \captionof{table}{Origin data.}
 		\begin{tabular}{ccc} 
 			\toprule
 			Record&R&H\\
 			\midrule
 			$\mathbf{p}_1$&15&102\\
 			$\mathbf{p}_2$&14&97\\
 			$\mathbf{p}_3$&20&99\\
 			$\mathbf{p}_4$&19&101\\
 			\bottomrule
 		\end{tabular}
 	\end{minipage}%
 	\begin{minipage}[t]{0.49\linewidth}
 		\centering
 		\caption{Mapped Database $\mathbf{T}$}
 		\label{table:MapData}
 		% \captionof{table}{Mapped data.}
 		\begin{tabular}{ccc}
 			\toprule 
 			Record&R&H\\
 			\midrule
 			$\mathbf{t}_1$&1&2\\
 			$\mathbf{t}_2$&2&3\\
 			$\mathbf{t}_3$&4&1\\
 			$\mathbf{t}_4$&3&1\\
 			\bottomrule
 		\end{tabular}
 	\end{minipage}
 
 \end{table}

 \begin{figure}[t!]
 	\centering
 	\includegraphics[width = 0.6\linewidth]{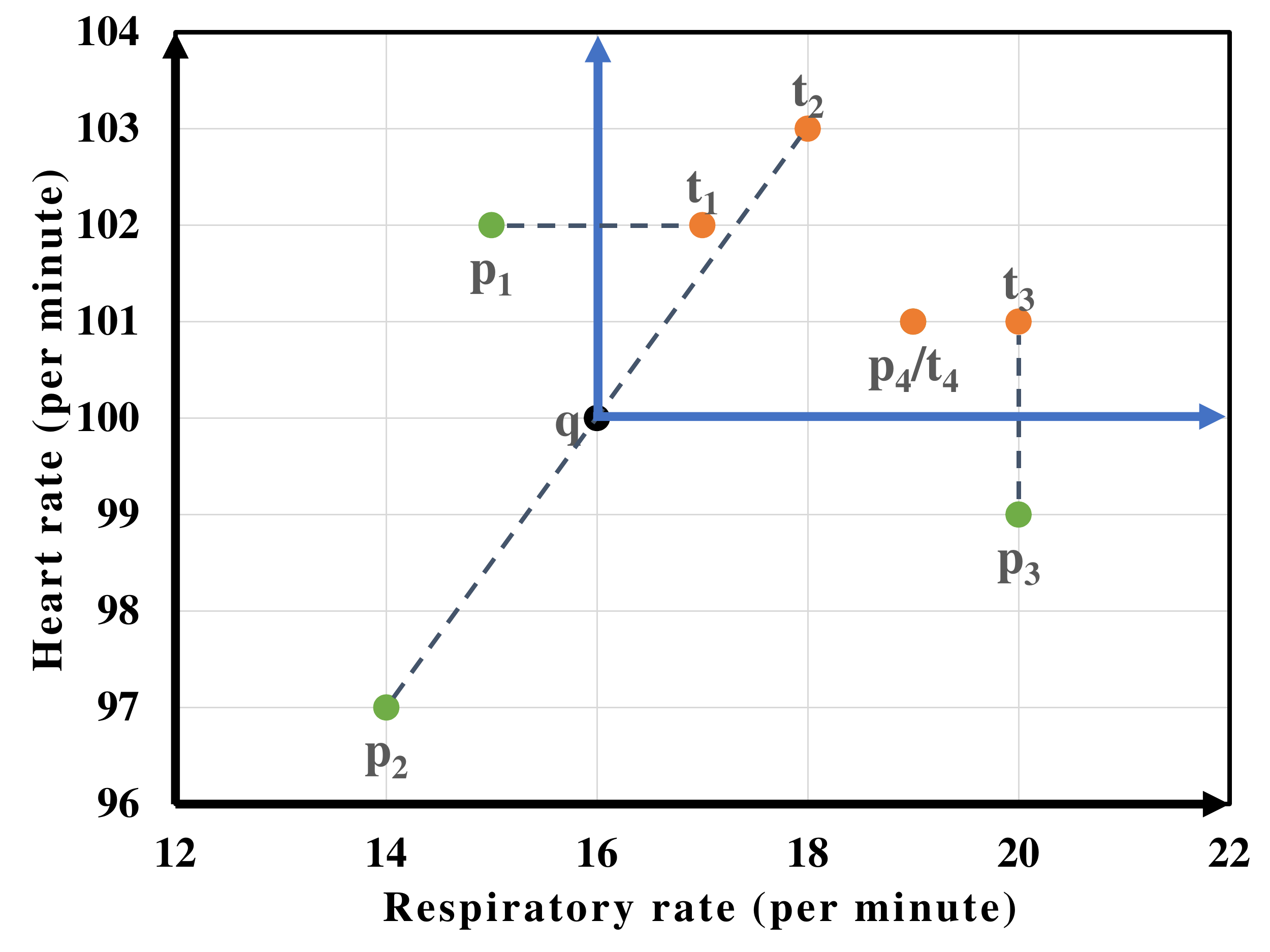}
 	\caption{An example of skyline query under a query point $\mathbf{q}$.}
 	\label{fig:Sample}
 \end{figure}

Specifically, we first consider how to support \emph{secure database mapping} given an encrypted query, allowing the cloud to securely map the encrypted outsourced database to the new space so as to facilitate the subsequent secure skyline tuples fetching. 
{\main} introduces an effective technique to tackle the challenging operation of computing absolute value in the secret sharing domain, securely realizing the operation of database mapping.
Then, {\main} introduces techniques to support \emph{secure skyline fetching}, allowing the cloud to obliviously fetch skyline tuples without knowing which tuples they are in the database.
After that, {\main} provides techniques for \emph{secure skyline and dominated tuples filtering} to tackle the remaining challenge, i.e., how to allow the cloud to obliviously filter out a  currently found skyline tuple and the tuples dominated by it from the mapped database without knowing which tuples they are and the number of dominated tuples.
The synergy of these secure components lead to the full protocol for fast privacy-preserving skyline queries developed in {\main}.

We implement our protocol and conduct extensive experiments on several datasets. 
The experiment results show that {\main} achieves substantial performance boost compared to the state-of-the-art works FSSP \cite{liu2019secure} and SMSQ \cite{ding2021efficient}.
Specifically, {\main} improves upon FSSP by up to $\mathbf{8130\times}$ and improves upon SMSQ by up to $\mathbf{813\times }$ in query latency.
We highlight our main contributions below:
\begin{itemize}
	\item  We present {\main}, a new system framework for secure skyline queries over encrypted databases outsourced to the cloud, which provides strong protection for the content confidentiality of the outsourced database, the query, and the result, as well as data patterns that may incur indirect data leakages.

	\item We devise a suite of secure and lightweight components to support oblivious skyline query processing at the cloud, including secure database mapping, secure skyline tuples fetching, and secure skyline and dominated tuples filtering.
		
	\item We formally analyze the security of {\main} and conduct extensive evaluations over several datasets.
	The results demonstrate that under the same system model and security guarantees, {\main} can achieve up to $813\times$ better query latency over the state-of-the-art \cite{ding2021efficient}, with promising scalability.
	
\end{itemize}

The rest of this paper is organized as follows.
Section \ref{sec:related} discusses the related work.
In Section \ref{sec:prelim}, we introduce preliminaries.
Then, we introduce our system architecture and threat model in Section \ref{sec:problem}.
After that, we present the design of {\main} in Section \ref{sec:design}, followed by security analysis and experiments in Section \ref{sec:security} and Section \ref{sec:expe}, respectively.
Finally, we conclude this paper in Section \ref{sec:conclusion}.

\section{Related Work}
\label{sec:related}
\subsection{Skyline Query in Plaintext Domain}

The skyline operator in the database filed is first proposed by B\"{o}rzs\"{o}nyi \textit{et al}. \cite{borzsonyi2001skyline}.
Since this seminal work, great efforts have been devoted to advancing the design of skyline query schemes. 
Kossmann \textit{et al}. \cite{kossmann2002shooting} study the online skyline using the nearest neighbor method.
Papadias \textit{et al}. \cite{papadias2005progressive} propose the branch and bound skyline algorithm, achieving performance boost in terms of efficiency and storage over  prior works.
The problem of skyline queries in different scenarios has also been widely studied, such as skyline on data streams \cite{tao2006maintaining}, uncertain skyline \cite{liu2015finding,pei2007probabilistic}, and group-based skyline \cite{liu2015findinga,yu2017fast}.
However, all of them consider the execution of skyline queries in the plaintext domain without considering privacy protection.

 \begin{algorithm}[t!]
            \caption{Skyline Query in Plaintext}
            \label{alg:skyline_plaintext}
            \begin{algorithmic}[1]
                \REQUIRE \revise{An} $m$-dimensional database $\mathbf{P}$ of $n$ tuples and a query tuple $\mathbf{q}$.
                \ENSURE The set of skyline tuples $\mathcal{SK}_{\mathbf{q}}$ with respect to $\mathbf{q}$.

                \FOR{$i=1$ to $n$}\label{alg:main_map_start}
                    \FOR{$j=1$ to $m$}
                        \STATE $\mathbf{t}_i[j]=|\mathbf{p}_i[j]-\mathbf{q}[j]|$.
                    \ENDFOR
                \ENDFOR
                
                \STATE Set  $\{\mathbf{t}_1, \cdots,\mathbf{t}_n\}$ as the initial mapped database $\mathbf{T}^{(0)}$. \label{alg:main_map_end}
                
                \FOR{$i=1$ to $n$}\label{alg:main_sum_begin}
                
                \STATE $\mathbf{s}[i]=\sum_{j=1}^{m}\mathbf{t}_i[j]$. 
                \ENDFOR\label{alg:main_sum_end}
                
                \STATE $k=0$.
                \WHILE{$\mathbf{T}^{(k)}\neq\emptyset$}
                    \STATE Select from the current mapped database $\mathbf{T}^{(k)}$ the tuple $\mathbf{t}_{i}$ with the minimum $\mathbf{s}[i]$, denoted by $\mathbf{t}_{\star}$. \label{alg:select_mim}
                    
                    \STATE Add the tuple in $\mathbf{P}$ corresponding to $\mathbf{t}_{\star}$ to the skyline pool $\mathcal{SK}_{\mathbf{q}}$. \label{alg:select_p_sky}

                    \STATE Delete $\mathbf{t}_{\star}$ and tuples dominated by $\mathbf{t}_{\star}$ from $\mathbf{T}^{(k)}$. \label{alg:elim_start}

                    \STATE $\mathbf{T}^{(k+1)}=\mathbf{T}^{(k)}$.
                    \STATE $k++$\label{alg:elim_end}.
                \ENDWHILE
                \RETURN $\mathcal{SK}_{\mathbf{q}}$.
            \end{algorithmic}
        \end{algorithm}
        
\subsection{Secure Skyline Query Processing}

% In recent decades, privacy-preserving database queries have been widely studied.
%
% However, most of existing works \cite{yuan2016building,liu2017secure,sun2018practical, dauterman2020dory,sun2016efficient,kamara2017boolean,lai2018result,patranabis2021forward} are focused on keyword search for different applications, which are different from the privacy-preserving skyline queries focused by us.
%

Bothe \textit{et al.} \cite{bothe2014skyline} initiate the first study on secure skyline query processing.
They introduce a preliminary approach that relies on a mechanism which multiplies vectors via secret matrices for protection.
Their approach does not provide formal and rigorous security guarantees.
Recently, Liu \textit{et al.} propose the FSSP scheme \cite{liu2019secure} (which first appeared in \cite{liu2017secure}), and Ding \textit{et al.} propose the SMSQ scheme \cite{ding2021efficient}.
Both these recent schemes provide strong cryptographic guarantees for the databases,the  skyline queries, as well as the query results.
However, as mentioned above, FSSP \cite{liu2019secure} and SMSQ \cite{ding2021efficient} rely on the use of heavy cryptosystems and incur substantial performance overheads, which heavily affect their practical usability.
In contrast, {\main} is a new system design for fast privacy-preserving skyline queries over encrypted databases hosted in the cloud, which fully builds on lightweight cryptography and achieves performance substantially better than the state-of-the-art prior schemes FSSP \cite{liu2019secure} and SMSQ \cite{ding2021efficient}.

%  on the lightweight cryptography---additive secret sharing \cite{demmler2015aby}.
% %
% Extensive experiments demonstrate that under the same system model and security guarantees, {\main} performs orders of magnitude better (\textit{several seconds versus hundreds of seconds})  than the existing works \cite{liu2019secure, ding2021efficient}.

% The works that are most related to ours are FSSP \cite{liu2019secure} and SMSQ \cite{ding2021efficient}, which provide strong cryptography guarantees for the database, skyline queries as well as the query results.
% %
% However, as mentioned above, the use of heavy cryptographic techniques results the poor performance of their schemes, which greatly degrades their practicability.
%

There are some works \cite{zhang2022efficient,wang2022efficient} focusing on privacy-preserving skyline query under application scenarios different from ours.
Specifically, the work \cite{zhang2022efficient} studies privacy-preserving user-defined skyline queries, focusing on a different and simplified case of constrained subspace skyline queries, where the client specifies a constrained region to search.
The underlying skyline query targeted in our security design as well as the prior works \cite{liu2019secure,ding2021efficient} is generic and much more challenging.
In addition, it is noted that the scheme in \cite{zhang2022efficient} does not offer protection for the access pattern.
Wang \textit{et al}. \cite{wang2022efficient} focus on the support for verifiability with respect to \textit{location-based} skyline queries where the client is only allowed to customize its skyline queries with two spatial attributes.
In addition, to achieve affordable online query latency, the scheme in \cite{wang2022efficient} requires the data owner to pre-compute the dominance relationships with respect to non-spatial attributes among database tuples before outsourcing them to the cloud.
%
% In contrast, {\main} allows the data owner to directly outsource the encrypted database to the cloud without pre-computation.
%
\revise{An additional work by Wang \textit{et al.} \cite{wang2020stargazing}
proposes a trusted hardware-based approach for privacy-preserving skyline query.
Such approach requires to put additional trust on trusted hardware vendors.
Moreover, in recent years, various attacks against trusted hardware have been proposed \cite{hahnel2017high,van2017telling,lee2017inferring,lee2020off}, which pose severe threats to trusted hardware-based secure systems, but the solution in \cite{wang2020stargazing} does not consider these attacks.
Hence, the state-of-the-art prior works that are most related to ours are \cite{liu2019secure,ding2021efficient}.}

\section{Preliminaries}
\label{sec:prelim}
\subsection{Skyline Query}
\label{subsec:Skyline}

    \begin{definition}
        \label{def:skyline}
       Given a database $\mathbf{P}=\{\mathbf{p}_1,\cdots,\mathbf{p}_n\}$, where each database tuple $\mathbf{p}_i$ ($i\in[1,n]$) is an $m$-dimensional vector, i.e., a tuple where each dimension corresponds to an attribute.
        Let $\mathbf{p}_a$ and $\mathbf{p}_b$ be two different tuples in $\mathbf{P}$.
        We say $\mathbf{p}_a$ dominates $\mathbf{p}_b$, if and only if $\forall j\in[1,m] $, $\mathbf{p}_a[j]\le\mathbf{p}_b[j]$ and $\exists j\in[1,m]$, $\mathbf{p}_a[j]<\mathbf{p}_b[j]$.
        Then the skyline tuples are tuples that are not dominated by any other tuple.
    \end{definition}

Given a query tuple, the skyline query targeted in this paper aims to retrieve from a database tuples that are not dominated by any other tuple \cite{DellisS07,liu2017secure}. 
\revise{The formal definition of skyline query considered in this paper is given below \cite{liu2019secure}}:

% a small set of skyline tuples which are not ``dominated'' by others from a large number of multi-dimensional data tuples, which can be formally defined as follows \cite{hose2012survey}.

% In this paper, we consider skyline queries in a practical dynamic scenario, where the skyline tuples are defined with respect to a query tuple. The formal definition of such dynamic skyline queries is as follows \cite{Zou2010Dynamic}.

% In this paper, we focus on the more advanced and practical skyline queries: dynamic skyline queries, which can be formally defined as follows \cite{Zou2010Dynamic}.
    \begin{definition}
        \label{def:dynSkyline}
        Given a query tuple $\mathbf{q}$ and a database $\mathbf{P}=\{\mathbf{p}_1,\cdots,\mathbf{p}_n\}$, where $\mathbf{q}$ has the same dimension as each tuple in $\mathbf{P}$. 
        Let $\mathbf{p}_a$ and $\mathbf{p}_b$ be two different tuples in $\mathbf{P}$.
        We say $\mathbf{p}_a$ dynamically dominates $\mathbf{p}_b$ with respect to $\mathbf{q}$, if and only if $~\forall j\in[1,m]$, $|\mathbf{p}_a[j]-\mathbf{q}[j]|\le|\mathbf{p}_b[j]-\mathbf{q}[j]|$, and $\exists ~ j\in[1,m]$, $|\mathbf{p}_a[j]-\mathbf{q}[j]|<|\mathbf{p}_b[j]-\mathbf{q}[j]|$.
        A skyline tuple with respect to $\mathbf{q}$ is a tuple that is not dominated by any other tuple. The set of skyline tuples under $\mathbf{q}$ is denoted by $\mathcal{SK}_{\mathbf{q}}$.

        % The skyline tuples with respect to $\mathbf{q}$ are those tuples that are not dynamically dominated by any other tuple, represented as a set $\mathcal{SK}_{\mathbf{q}}$.
    \end{definition}

         Algorithm \ref{alg:skyline_plaintext} shows the plaintext-domain processing of the skyline query \cite{liu2017secure,liu2019secure}.
        Given a database $\mathbf{P}$ and a query tuple $\mathbf{q}$, the first step is to map the database $\mathbf{P}$ to a new database (referred to as \textit{mapped database}) with respect to $\mathbf{q}$ (i.e., lines \ref{alg:main_map_start}-\ref{alg:main_map_end}).
        Then, for each tuple in the initial mapped database (i.e., $\mathbf{T}^{(0)}$), the sum over its all attributes (i.e., lines \ref{alg:main_sum_begin}-\ref{alg:main_sum_end}) is computed.
        Skyline tuples are selected from the mapped database in turn through multiple rounds.
        In the $k$-th round ($k \geq 0$), the current mapped database $\mathbf{T}^{(k)}$ is taken as input, and a skyline tuple $\mathbf{t}_{\star}$ in $\mathbf{T}^{(k)}$ is chosen, which is the one with the smallest attribute sum. 
        The tuple in the original database $\mathbf{P}$ corresponding to $\mathbf{t}_{\star}$ is added to the skyline pool.
        After that, the skyline tuple $\mathbf{t}_{\star}$ and those tuples dominated by it are deleted from $\mathbf{T}^{(k)}$, producing the updated mapped database $\mathbf{T}^{(k+1)}$ to be used in the next round.
        % We name the $k$-th round as the \textit{current round}, the $(k-1)$-th round as the \textit{previous round}, and the $(k+1)$-th round as the \textit{next round}.
        %
        % The current round inputs the mapped database $\mathbf{T}^{(k)}$ (named as \textit{current mapped database}) output from the previous round, and then outputs the new mapped database $\mathbf{T}^{(k+1)}$ for the use of skyline tuple selection in the next round.
        %
        % The skyline tuple in $\mathbf{T}^{(k)}$ is the tuple with the smallest sum of attributes. %  (i.e., line \ref{alg:select_mim}).
        %
        % After that, the skyline tuple and those tuples dominated by the skyline tuple are deleted from $\mathbf{T}^{(k)}$, outputting $\mathbf{T}^{(k+1)}$.
        %
        The process is repeated through multiple rounds until the mapped database becomes empty.

    \subsection{Additive Secret Sharing}
    %We use the additive secret sharing \cite{mohassel2017secureml} as the basic building block in our secure skyline protocol.
    %
    \label{sec:Sharing}
    
    Additive secret sharing \cite{demmler2015aby} is a lightweight encryption technique that allows some secure computation.
    Given a secret value $x \in \mathbb{Z}_{2^l}$, additive secret sharing in a two-party setting works by splitting it into two secret shares $\langle x\rangle^A_1 \in \mathbb{Z}_{2^l}$ and $\langle x\rangle^A_2\in \mathbb{Z}_{2^l}$.
    For $l>1$, $ x=\langle x\rangle^A_1+\langle x\rangle^A_2$ in $\mathbb{Z}_{2^l}$ and such sharing is referred to as arithmetic sharing.
    For $l=1$, $ x=\langle x\rangle^A_1\oplus \langle x\rangle^A_2$ in $\mathbb{Z}_{2}$, and such sharing is referred to as binary sharing.
    Each share alone reveals not information about $x$.
    The shares are to be held by two parties $P_{1}$ and $P_{2}$ respectively for subsequent secure computation.
    We write $\llbracket x\rrbracket^A$ and $\llbracket x\rrbracket^B$ respectively to clearly distinguish between arithmetic sharing and binary sharing in the above form.

    With the shares of two secret values $x$ and $y$ held by two parties $P_{1}$ and $P_{2}$ respectively, some operations can be performed securely among them. We use arithmetic sharing to illustrate the secure computation.
    Note that in binary sharing, the only differences are that addition/subtraction operations are replaced by XOR ($\oplus$) and multiplication operations are replaced by AND ($\otimes$).

    In particular, the addition/subtraction between two secret-shared values $\llbracket x\rrbracket^{A}$ and $\llbracket y\rrbracket^{A}$ only requires local computation at each party, i.e., $\langle z\rangle^A_i=\langle x\rangle^A_i\pm\langle y\rangle^A_i,i\in \{1,2\}$. Also, the scalar multiplication between a public value $\eta$ and a secret-shared value $\llbracket x\rrbracket^{A}$ also only requires local computation,  i.e., $\langle z\rangle^A_i=\eta\cdot\langle x\rangle^A_i$. 
    The multiplication between two secret-shared values $\llbracket x\rrbracket^{A}$ and $\llbracket y\rrbracket^{A}$, however, requires one round of online communication. 
    Specifically, to compute $\llbracket z\rrbracket^A$ where $z=xy$, $P_{1}$ and $P_{2}$ need to additionally have as input a secret-shared Beaver triple $(\llbracket u\rrbracket^{A}, \llbracket v\rrbracket^{A},\llbracket w\rrbracket^{A} )$ which can be prepared offline \cite{RiaziWTS0K18}, where $w=uv$.
    Then, each party first locally computes $\langle e\rangle_i=\langle x\rangle_i-\langle u\rangle_i$, $\langle f\rangle_i=\langle y\rangle_i-\langle v\rangle_i$, and then reveal $e$ and $f$ to each other.
    Finally, $P_1$ and $P_{2}$ locally compute the secret shares of $z$ by $\langle z\rangle_1= e\cdot f+f\cdot\langle u\rangle_1+e\cdot\langle v\rangle_1+\langle w\rangle_1$ and $\langle z\rangle_2= f\cdot\langle u\rangle_2+e\cdot\langle v\rangle_2+\langle w\rangle_2$, respectively.
    For simplicity, we write $\llbracket z\rrbracket^A=\llbracket x\rrbracket^A\cdot\llbracket y\rrbracket^A$ to denote such secure multiplication. 
    \revise{In addition, the NOT operation (denoted by $\neg$) in binary 
    secret sharing domain can be realized by letting one of $P_{1}$ and $P_{2}$ locally flip the share it holds, e.g., $\langle\neg x\rangle^B_1=\neg\langle x\rangle^B_1,\langle\neg x\rangle^B_2=\langle x\rangle^B_2$.
    }

\section{Problem Statement}
\label{sec:problem}
    \subsection{System Architecture}

    Fig. \ref{fig:Arch}  illustrates the system architecture of {\main}.
    There are three kinds of entities: the data owner, the client, and the cloud.
    The data owner can be an organization (e.g., a medical institution), who has a database $\mathbf{P}$ and wants to offer skyline query services to clients (e.g., doctors in a hospital). 
    To leverage the well-known benefits of cloud computing \cite{QinW0018,JiangWHWLSR21}, the data owner intends to store the database $\mathbf{P}$ in the cloud, who then helps provide skyline query services for the client. 
   	Due to privacy concerns, it is demanded that security must be embedded in such cloud-empowered service, safeguarding the database $\mathbf{P}$, skyline query $\mathbf{q}$, as well as the corresponding query result $\mathcal{SK}_{\mathbf{q}}$.

    % However, the cloud may learn the private information from $\mathbf{P}$, e.g., medical history. 
    %
    % Therefore, considering the data breaches \cite{RenWW12} and that the database $\mathbf{P}$ is proprietary, it is demanded that security must be embedded in such cloud service, safeguarding the database $\mathbf{P}$, dynamic skyline query requests $\mathbf{q}$ as well as the corresponding query results $\mathcal{SK}_{\mathbf{q}}$. 
    %

    For high efficiency, in {\main} we resort to a lightweight cryptographic technique---additive secret sharing---for fast encryption of the database and skyline query and for supporting subsequent secure processing in the cloud, through a customized design.
    To be compatible with the working paradigm of additive secret sharing, the power of the cloud in {\main} is divided into two cloud servers (denoted by $C_{1}$ and $C_{2}$) who can be hosted by independent cloud service providers, e.g., Google, AWS, and Microsoft in practice. 
    Such a two-server model has also been adopted in state-of-the-art prior works on privacy-preserving skyline queries \cite{ding2021efficient,liu2017secure,liu2019secure}, as well as in other application domains \cite{mohassel2017secureml,meng2018top,chen2020metal,zheng2022optimizing,du2020graphshield,cui2020svknn,wang2022privacy,zheng2020securely,wang2022PeGraph}.
    \revise{In addition to the adoption in academia, the two-server model has also gained increasing traction in industry.
    For example, Mozilla initiates a secure telemetry data collection service on Firefox under the two-server model \cite{Mozilla}; Apple and Google collaboratively provide users with automated alerts about potential COVID-19 exposure, while providing strong privacy guarantees \cite{WhitePaper}.}
    {\main} follows such trend and contributes a new design for enabling fast privacy-preserving skyline queries over encrypted cloud databases. 

        \begin{figure}[t!]
    	\centering
    	\includegraphics[width = \linewidth]{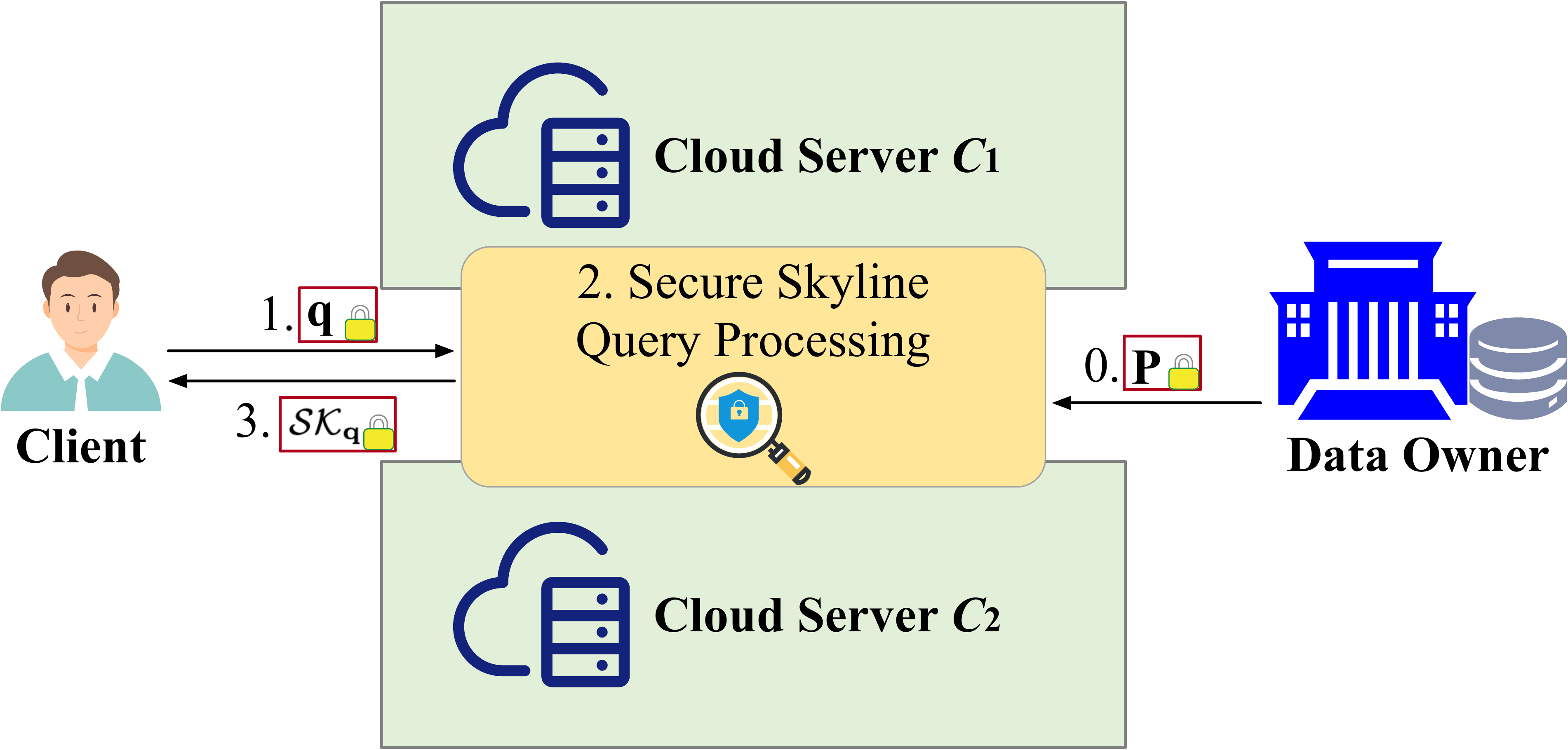}
    	\caption{The architecture of {\main}.}
    	\label{fig:Arch}
    \end{figure}

    \subsection{Threat Model}
    
    Similar to the state-of-the-art prior works on privacy-preserving skyline queries \cite{liu2019secure,ding2021efficient} as well as other works in the two-server setting \cite{du2020graphshield,chen2020metal,WangWHZR16,0002SKG19}, we assume a semi-honest and non-colluding adversary model where each cloud server honestly follows our protocol, yet may \textit{individually} attempt to learn the private information from the execution of (dynamic) skyline queries.
    Following the prior works \cite{liu2019secure,ding2021efficient}, we consider the data owner and the client as trustworthy parties, who will honestly follow the protocol specification.
   
    Under the above threat model and following the state-of-the-art prior works \cite{liu2019secure,ding2021efficient}, {\main} aims to protect against the cloud servers (i) the content of the database $\mathbf{P}$, skyline query $\mathbf{q}$, and query result $\mathcal{SK}_{\mathbf{q}}$, (ii) the dominance relationships among database tuples, (iii) the number of database tuples that each skyline tuple dominates, and (iv) search access patterns.
    \revise{Following the standard definitions in searchable encryption \cite{curtmola2006searchable}, we describe the search access patterns in secure skyline queries as follows.}

    \revise{
    \begin{definition}
        \textbf{Search pattern.} For two skyline queries $\mathbf{q}$ and $\mathbf{q}^{\prime}$,
        define $\Sigma(\mathbf{q},\mathbf{q}^{\prime})\in\{0,1\}$, where $\Sigma(\mathbf{q},\mathbf{q}^{\prime})=1$
        if and only if the two queries are identical, and otherwise $\Sigma(\mathbf{q},\mathbf{q}^{\prime})=0$.
        Here, ``identical'' means that all corresponding attribute values of $\mathbf{q}$ and $\mathbf{q}^{\prime}$ are identical.
        Let $\mathbf{Q}=\{\mathbf{q}_1,\cdots,\mathbf{q}_r\}$ be a non-empty sequence of skyline queries.
        The search pattern reveals an $r\cdot r$
        (symmetric) matrix with element $(i,j)$ equal to $\Sigma(\mathbf{q}_i,\mathbf{q}_j)$.
    \end{definition}
    In short, the search pattern implies whether a new skyline query has been issued before.}

    \revise{
    \begin{definition}
        \textbf{Access pattern.} Given a skyline query $\mathbf{q}$ on the database $\mathbf{P}$,
        the access pattern reveals the indexes of skyline tuples with respect to $\mathbf{q}$ in $\mathbf{P}$.
    \end{definition}
    In practice, the access pattern reveals which database tuples are the skyline tuples with respect to a given query.}

\section{The Design of {\main}}
\label{sec:design}
    \subsection{Overview}

    At a high level, the protocol in {\main} proceeds through the following phases.
    Firstly, in an initialization phase, the data owner adequately encrypts each tuple in its database $\mathbf{P}$ under arithmetic additive secret sharing and produces $\llbracket\mathbf{P}\rrbracket^A$. The data owner then sends the secret shares $\langle\mathbf{P}\rangle^A_{1}$ and $\langle\mathbf{P}\rangle^A_{2}$ to $C_{1}$ and $C_{2}$, respectively.
	Subsequently, it comes to the online query phase, where the client first encrypts its skyline query tuple $\mathbf{q}$ through arithmetic sharing and sends the secret shares  $\langle\mathbf{q}\rangle^A_{1}$ and $\langle\mathbf{q}\rangle^A_{2}$ to the cloud servers $C_{1}$ and $C_{2}$, respectively.
    Hereafter, for simplicity of presentation, we will write {\csa} to represent the two cloud servers $C_{1}$ and $C_{2}$.
    Upon receiving the encrypted query $\llbracket\mathbf{q}\rrbracket^A$, {\csa} securely process the encrypted skyline query over $\llbracket\mathbf{P}\rrbracket^A$ as per the customized design of {\main}.

    To allow {\csa} to perform the skyline query processing (i.e., Algorithm \ref{alg:skyline_plaintext}) in an oblivious manner, we first conduct an in-depth examination on the whole procedure and decompose it into several essential components, for which we provide customized secure realizations.
    Specifically, we identify and devise the following secure components for supporting secure skyline queries.
    \begin{itemize}
        \item \textit{Secure database mapping $\mathsf{secMap}$}. % (line \ref{alg:main_map_start}-\ref{alg:main_map_end} in Algorithm \ref{alg:skyline_plaintext}). 
        Given the encrypted database $\llbracket\mathbf{P}\rrbracket^A$ and query $\llbracket\mathbf{q}\rrbracket^A$, {\main} provides $\mathsf{secMap}$ to have {\csa} securely map the encrypted database $\llbracket\mathbf{P}\rrbracket^A$ to the encrypted mapped database  $\llbracket\mathbf{T}\rrbracket^A$ with respect  to $\llbracket\mathbf{q}\rrbracket^A$ so as to facilitate the subsequent secure skyline tuples fetching.

        \item \textit{Secure skyline fetching $\mathsf{secFetch}$}. % (line \ref{alg:select_mim}-\ref{alg:select_p_sky} in Algorithm \ref{alg:skyline_plaintext}). 
        Given the current encrypted mapped database, {\main} provides $\mathsf{secFetch}$ to allow {\csa} to obliviously fetch the skyline tuple $\llbracket \mathbf{t}_{\star}\rrbracket^A$ from the current mapped database without knowing which tuple it is.
        Meanwhile, $\mathsf{secFetch}$ allows {\csa} to obliviously fetch the skyline tuple $\llbracket \mathbf{p}_{\star}\rrbracket^A$ corresponding to $\llbracket \mathbf{t}_{\star}\rrbracket^A$ from $\llbracket\mathbf{P}\rrbracket^A$, which is added into the encrypted skyline pool $\llbracket\mathcal{SK}_{\mathbf{q}}\rrbracket^A$.

        \item \textit{Secure skyline and dominated tuples filtering  $\mathsf{secFilt}$}. % (line \ref{alg:elim_start} in Algorithm \ref{alg:skyline_plaintext}). 
        Given  $\llbracket\mathbf{T}\rrbracket^A$ and $\llbracket \mathbf{t}_{\star}\rrbracket^A$, {\main} provides $\mathsf{secFilt}$ to allow {\csa} to obliviously filter out $\llbracket \mathbf{t}_{\star}\rrbracket^A$ and the tuples dominated by $\llbracket \mathbf{t}_{\star}\rrbracket^A$ without knowing which tuples they are as well as the number of the dominated tuples.
    
    \end{itemize}

    \begin{algorithm}[t!]
        \caption{Secure Database Mapping $\mathsf{secMap}$}
        \label{alg:Mapping}
        \begin{algorithmic}[1]
            \REQUIRE The encrypted original database $\llbracket{\mathbf{P}}\rrbracket^A$ and skyline query $\llbracket\mathbf{q}\rrbracket^A$.
            
            \ENSURE The encrypted mapped database $\llbracket\mathbf{T}\rrbracket^A$.
            
            \STATE Initialization: $\llbracket\mathbf{T}\rrbracket^A=\emptyset$.
            \FOR{$i=1$ to $n$}
                \FOR{$j=1$ to $m$}

                    \STATE $\llbracket b\rrbracket^B=\mathsf{SecExt}(\llbracket\mathbf{p}_i[j]\rrbracket^A,\llbracket\mathbf{q}[j]\rrbracket^A)$.\label{line:MapMSB}

                   % \STATE Compute NOT: $\langle b^\prime\rangle^B_1=\neg\langle b\rangle^B_1,\langle b^\prime\rangle^B_2=\langle b\rangle^B_2$.
                    
                    \STATE $\llbracket b'\rrbracket^B=\llbracket \neg b\rrbracket^B$.\label{line:MapNOT}

                    \STATE $\llbracket\mathbf{t}_i[j]\rrbracket^A=\mathsf{MultiBA}(\llbracket b\rrbracket^B,\llbracket\mathbf{q}[j]\rrbracket^A-\llbracket\mathbf{p}_i[j]\rrbracket^A)+\mathsf{MultiBA}(\llbracket b^\prime\rrbracket^B,\llbracket\mathbf{p}_i[j]\rrbracket^A-\llbracket\mathbf{q}[j]\rrbracket^A)$.\label{line:MapAMB}
                \ENDFOR
                \STATE $\llbracket\mathbf{T}\rrbracket^A.append(\llbracket\mathbf{t}_i\rrbracket^A)$.
            \ENDFOR
            \RETURN $\llbracket\mathbf{T}\rrbracket^A$.
        \end{algorithmic}
    \end{algorithm}

    Next, we will introduce the detailed design of $\mathsf{secMap}$ in Section \ref{subsec:Mapping}, $\mathsf{secFetch}$ in Section \ref{subsec:find}, and $\mathsf{secFilt}$ in Section \ref{subsec:elim}.
   Afterwards, in Section \ref{subsec:secure_main}, we give the complete protocol in {\main} for secure skyline query processing at the cloud, which relies on the synergy of the three secure components devised in {\main}.

    \subsection{Secure Database Mapping} 

    \label{subsec:Mapping}

   Secure database mapping $\mathsf{secMap}$ aims at allowing {\csa} to securely map the encrypted database $\llbracket\mathbf{P}\rrbracket^A$ to the encrypted mapped database  $\llbracket\mathbf{T}\rrbracket^A$ with respect to the encrypted query $\llbracket\mathbf{q}\rrbracket^A$. 
   From the process in Algorithm \ref{alg:skyline_plaintext}, i.e., lines \ref{alg:main_map_start}-\ref{alg:main_map_end}, we observe that the challenge here is to securely calculate the absolute value $\mathbf{t}_i[j]=|\mathbf{p}_i[j]-\mathbf{q}[j]|$ in the secret sharing domain. 
   Therefore, we design a tailored protocol to allow {\csa} to securely evaluate the encrypted  absolute value $\llbracket |a-b|\rrbracket^A$ when they hold the secret sharings $\llbracket a\rrbracket^A$ and $\llbracket b\rrbracket^A$.

    Our solution to this challenge is based on the following observation:
    \begin{equation}\label{eq:abs}
   	|a-b|=(a<b) \cdot (b-a)+\neg(a<b)\cdot (a-b),
    \end{equation}
    % \begin{equation}\label{eq:abs}
    % \llbracket	|a-b|\rrbracket^A=\llbracket(a<b)\rrbracket^B\cdot\llbracket b-a\rrbracket^A+\llbracket\neg(a<b)\rrbracket^B\cdot\llbracket a-b\rrbracket^A,
    % \end{equation}
    where $\neg$ represents the NOT operation and $(a<b)=1 \in \mathbb{Z}_{2}$ if $a<b$, and $(a<b)=0 \in \mathbb{Z}_{2}$ if $a\geq b$.
    % $\neg$ represents NOT operation and $(a<b)=1$ (in binary) if $a<b$, and $(a<b)=0$  if $a\geq b$.
    %
    %
    Given this observation, what needs to be considered is how to securely realize the computation of $(a<b)$ as well as the NOT operation in the secret sharing domain.
    \revise{As mentioned in Section \ref{sec:Sharing}, the NOT operation on a secret-shared bit can be simply achieved by letting one of {\csa} ($C_1$ undertakes this in {\main}) locally flip the share it holds.}
    So it remains to be considered how to allow {\csa} to securely evaluate $a<b$ with the secret sharings $\llbracket a\rrbracket^A$ and $\llbracket b\rrbracket^A$.

    % For the NOT operation in the secret sharing domain, we can let one of {\csa} locally flip the share it holds.
    %
    % Then, the next challenge is how to allow {\csa} to securely evaluate $\llbracket(a<b)\rrbracket^B$ when they hold the secret sharings $\llbracket a\rrbracket^A$ and $\llbracket b\rrbracket^A$.

    Here we resort to the strategy of secure bit decomposition in the secret sharing domain \cite{liu2021medisc,mohassel2018aby}. 
    Specifically, given $a,b\in \mathbb{Z}_{2^l}$ under two's complement representation, the most significant bit (MSB) of $a-b$ (denoted as $msb(a-b)$) can indicate whether $a<b$ or not.
    Namely, if $a-b<0$, $msb(a-b)=1$ and otherwise $msb(a-b)=0$.
    Secure extraction of the MSB in the secret sharing domain can be achieved by securely realizing a parallel prefix adder (PPA), which only requires basic $\oplus$ and $\otimes$ operations in the secret sharing domain.
    Fig. \ref{fig:PPA} illustrates an 8-bit PPA for MSB extraction.
    In \cite{liu2021medisc}, a concrete construction for secure MSB extraction based on PPA was provided, which allows two parties holding the secret sharings of two values $a$ and $b$ as input to obtain the secret sharing of the MSB of $a-b$.
    Let $\mathsf{SecExt}$ denote the secure MSB extraction protocol, for which we have $\llbracket msb(a-b)\rrbracket^B=\mathsf{SecExt}(\llbracket a\rrbracket^A,\llbracket b\rrbracket^A)$.
    For more details on the construction $\mathsf{SecExt}$, we refer the readers to \cite{liu2021medisc}.
    % To securely leverage the PPA for MSB extraction of $(a-b)$, $C_{1}$ inputs $\langle a  \rangle_{1}-\langle b \rangle_{1}$  and $C_{2}$ inputs $\langle a  \rangle_{2}-\langle b \rangle_{2}$.
    %
    % Then the secure tailored PPA outputs the binary secret sharings $\llbracket msb(a-b)\rrbracket^B$.
    %     
    % Recall the operations of binary secret sharing described in Section \ref{sec:Sharing}, $\otimes$ gate only requires 4 bits online communication in one round but $\oplus$ gate requires no communication.
    %
    % Therefore, the secure MSB extraction via an $l$-bit tailored PPA only requires the two cloud severs to online communicate $12l-12-4\log l$ bits in $\lceil\log l\rceil$ rounds.
    %
    It is noted that the output $\llbracket msb(a-b)\rrbracket^B$ from $\mathsf{SecExt}$ is in binary secret sharing domain.
    However, we need to obtain $\llbracket |a-b|\rrbracket^A$ as the result according to the computation in Eq. \ref{eq:abs}.
    So we need to consider how to perform multiplication between $\llbracket msb(a-b)\rrbracket^B$ (i.e., $\llbracket (a<b)\rrbracket^B$) and $\llbracket b-a\rrbracket^A$ as well as $\llbracket \neg msb(a-b)\rrbracket^B$ (i.e., $\llbracket\neg (a<b)\rrbracket^B$) and $\llbracket a-b\rrbracket^A$.
    That is, given the secret sharings $\llbracket x \rrbracket^B$ and $\llbracket y \rrbracket^A$, we want to obtain $\llbracket x\cdot y \rrbracket^A$.
    Inspired by \cite{mohassel2018aby}, {\main} deals with the multiplication of secret-shared values in different domains as follows.
    \revise{\begin{enumerate}
        \item $C_1$ draws a random value $r_1\in\mathbb{Z}_{2^l}$ and constructs two messages: $m_\mu:=(\mu\oplus\langle x \rangle^B_1)\cdot\langle y\rangle^A_1-r_1,\mu\in\{0,1\}$, and then sends $m_0,m_1$ to $C_2$.
        \item $C_2$ chooses $m_\mu$ according to the secret share $\langle x \rangle^B_2$ it holds. That is, $C_2$ chooses $m_0$ if $\langle x\rangle^B_2=0$ and $C_2$ chooses $m_1$ if $\langle x\rangle^B_2=1$. Then, $C_2$ holds the intermediate value $x\cdot\langle y\rangle^A_1-r_1$ and $C_1$ holds $r_1$.
        \item For the secret share $\langle y\rangle^A_2$, $C_2$ acts as the sender and $C_1$ acts as the receiver to repeat step 1) and 2). Then $C_1$ holds the intermediate value $x\cdot\langle y\rangle^A_2-r_2$ and $C_2$ holds the random value $r_2$ it draws.
        \item $C_1$ and $C_2$ respectively compute the shares of $\llbracket x\cdot y\rrbracket^A$ by: $\langle x\cdot y\rangle^A_1=r_1+x\cdot\langle y\rangle^A_2-r_2$, $\langle x\cdot y\rangle^A_2=r_2+x\cdot\langle y\rangle^A_1-r_1$. 
        It is easy to see that $\langle x\cdot y\rangle^A_1+\langle x\cdot y\rangle^A_2=x\cdot(\langle y\rangle^A_1+\langle y\rangle^A_2)=x\cdot y$.
    \end{enumerate}}
    Finally, $C_1$ and $C_2$ can obtain the secret sharing $\llbracket x \cdot y\rrbracket^A$.
    Let $\mathsf{MultiBA}$ denote such secret-shared multiplication, for which we have $\llbracket x\cdot y \rrbracket^A=\mathsf{MultiBA}(\llbracket x \rrbracket^B,\llbracket y \rrbracket^A)$.
    Analogously, $\mathsf{MultiBA}$ can also be applied on the secret-shared component-wise multiplication between a binary secret-shared value $\llbracket x\rrbracket^B$ and an arithmetic secret-shared vector $\llbracket\mathbf{v}\rrbracket^A$, for which we have  $\llbracket x\cdot\mathbf{v}\rrbracket^A=\mathsf{MultiBA}(\llbracket x\rrbracket^B,\llbracket\mathbf{v}\rrbracket^A)$, where $x\cdot\mathbf{v}$ is a vector from component-wise multiplication.
    With the above secure operations, we present the details of secure database mapping in Algorithm \ref{alg:Mapping}.

      \begin{figure}[t!]
        \centering
        \includegraphics[width = 0.8\linewidth]{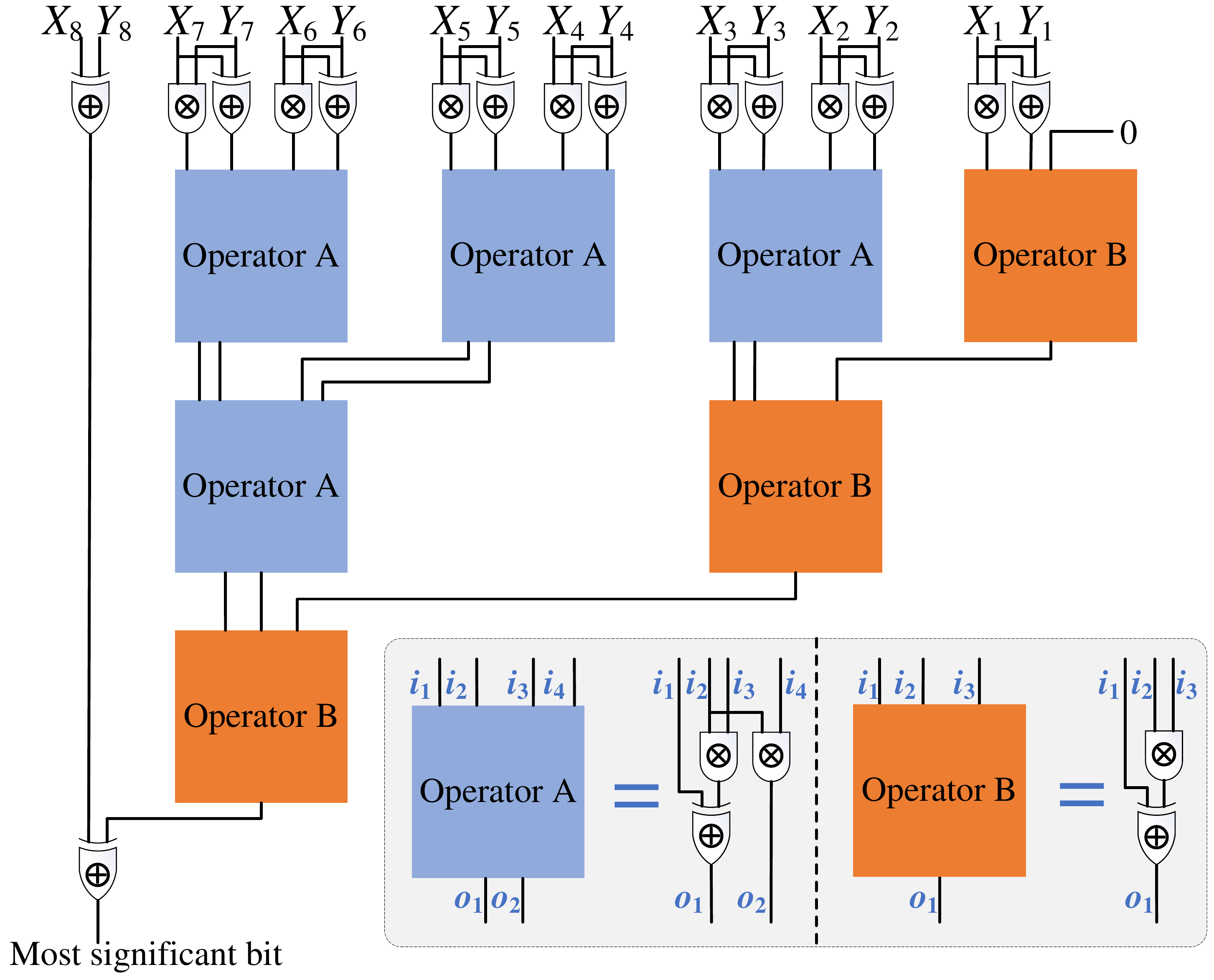}
        \caption{An 8-bit PPA for MSB extraction.}
        \label{fig:PPA}
    \end{figure}

     \subsection{Secure Skyline Fetching} 
       \label{subsec:find}
    After mapping the encrypted database $\llbracket{\mathbf{P}}\rrbracket^A$ to $\llbracket\mathbf{T}\rrbracket^A$ with respect to $\llbracket\mathbf{q}\rrbracket^A$, {\csa} need to obliviously fetch the skyline tuples $\{\llbracket\mathbf{t}_{\star}\rrbracket^A\}$ from $\llbracket\mathbf{T}\rrbracket^A$ and the skyline tuples $\{\llbracket\mathbf{p}_{\star}\rrbracket^A\}$ corresponding to $\{\llbracket\mathbf{t}_{\star}\rrbracket^A\}$ from $\llbracket\mathbf{P}\rrbracket^A$.
    For simplicity of presentation,  we next introduce how to allow  {\csa} to obliviously fetch one $\llbracket\mathbf{t}_{\star}\rrbracket^A$ from $\llbracket\mathbf{T}\rrbracket^A$ and its corresponding $\llbracket\mathbf{p}_{\star}\rrbracket^A$ from $\llbracket\mathbf{P}\rrbracket^A$.

   	According to the plaintext-domain shown in Algorithm \ref{alg:skyline_plaintext}, secure skyline fetching first needs to compute the attribute sum for each tuple $\llbracket\mathbf{t}_{i}\rrbracket^A \in \llbracket\mathbf{T}\rrbracket^A$.
    The summation operation is naturally supported in the secret sharing domain, namely,
    \begin{equation}\notag
    	\llbracket  \mathbf{s}[i]\rrbracket^A=\sum_{j=1}^{m}\llbracket\mathbf{t}_{i}[j]\rrbracket^A,
    \end{equation}  
     where $\llbracket \mathbf{s}[i]\rrbracket^A$ represents the attribute sum for tuple $\llbracket\mathbf{t}_{i}\rrbracket^A$.
     With this, {\main} devises a component $\mathsf{secFetch}$ to allow {\csa} to obliviously fetch the skyline tuple $\llbracket\mathbf{t}_{\star}\rrbracket^A$ from $\llbracket\mathbf{T}\rrbracket^A$ and its corresponding $\llbracket\mathbf{p}_{\star}\rrbracket^A$ from $\llbracket\mathbf{P}\rrbracket^A$.
     Note that $\mathbf{t}_{\star}$ refers to the tuple which has the minimum attribute sum in $\mathbf{T}$.
    Therefore, the first challenge of secure skyline fetching in the secret sharing domain is how to allow {\csa} to obliviously fetch the minimum value from a set of secret-shared values without knowing which and what value it is.

    \begin{algorithm}[t!]
    	\caption{Secure Skyline Fetching $\mathsf{secFetch}$}
    	\label{alg:bFind}
    	\begin{algorithmic}[1]
    		\REQUIRE The encrypted sum vector $\llbracket\mathbf{s}\rrbracket^A$, encrypted mapped database $\llbracket\mathbf{T}\rrbracket^A$, and encrypted original database $\llbracket\mathbf{P}\rrbracket^A$.
    		\ENSURE The encrypted minimum $\llbracket  sMin\rrbracket^A\in \llbracket\mathbf{s}\rrbracket^A$,  and the encrypted skyline tuples $\llbracket\mathbf{t}_{\star}\rrbracket^A$ and  $\llbracket\mathbf{p}_{\star}\rrbracket^A$. % in $\llbracket\mathbf{T}\rrbracket^A$ and $\llbracket\mathbf{P}\rrbracket^A$, respectively.
    		
    		\STATE Initialization: $\llbracket  sMin\rrbracket^A=\llbracket\mathbf{s}[1]\rrbracket^A$, $\llbracket\mathbf{t}_{\star}\rrbracket^A=\llbracket\mathbf{t}_{1}\rrbracket^A\in \llbracket\mathbf{T}\rrbracket^A$, $\llbracket\mathbf{p}_{\star}\rrbracket^A=\llbracket\mathbf{p}_{1}\rrbracket^A \in \llbracket\mathbf{P}\rrbracket^A$.
    		
    		\FOR{$i=2$ to $n$}
    		\STATE $\llbracket\varphi\rrbracket^B=\mathsf{SecExt}(\llbracket\mathbf{s}[i]\rrbracket^A,\llbracket sMin\rrbracket^A)$.\label{line:FindMSB}
    		
    		\STATE $\llbracket\varphi'\rrbracket^B=\llbracket\neg\varphi\rrbracket^B$.
    		
    		\STATE $\llbracket  sMin\rrbracket^A=\mathsf{MultiBA}(\llbracket\varphi\rrbracket^B,\llbracket\mathbf{s}[i]\rrbracket^A)+$\\$\mathsf{MultiBA}(\llbracket\varphi^\prime\rrbracket^B,\llbracket sMin\rrbracket^A)$.\label{line:FindAMB1}
    		\STATE $\llbracket\mathbf{t}_{\star}\rrbracket^A=\mathsf{MultiBA}(\llbracket\varphi\rrbracket^B,\llbracket\mathbf{t}_{i}\rrbracket^A)+$\\$\mathsf{MultiBA}(\llbracket\varphi^\prime\rrbracket^B,\llbracket\mathbf{t}_{\star}\rrbracket^A)$.
    		\STATE $\llbracket\mathbf{p}_{\star}\rrbracket^A=\mathsf{MultiBA}(\llbracket\varphi\rrbracket^B,\llbracket\mathbf{p}_{i}\rrbracket^A)+$\\$\mathsf{MultiBA}(\llbracket\varphi^\prime\rrbracket^B,\llbracket\mathbf{p}_{\star}\rrbracket^A)$.\label{line:FindAMB3}
    		\ENDFOR
    		\RETURN $\llbracket  sMin\rrbracket^A$, $\llbracket\mathbf{t}_{\star}\rrbracket^A$, and $\llbracket\mathbf{p}_{\star}\rrbracket^A$.
    	\end{algorithmic}
    \end{algorithm}

    Obviously finding the minimum value from several values essentially needs comparison followed by swapping based on the comparison.
    This can be securely realized as follows.
    % Therefore, we design a mechanism to allow {\csa} to obliviously perform the two operations in the secret sharing domain.
    %
    Firstly, given the secret sharings $\llbracket a\rrbracket^A$ and $\llbracket b\rrbracket^A$ held by  $C_1$ and $C_2$, we can first leverage $\mathsf{SecExt}$ to obtain the secret-shared result $\llbracket\varphi\rrbracket^B$ of comparison between $a$ and $b$, i.e., $\llbracket\varphi\rrbracket^B=\mathsf{SecExt}(\llbracket a\rrbracket^A,\llbracket b\rrbracket^A)$.
    Then, the smaller value among $a$ and $b$ can be obliviously selected via 
    \begin{equation}\notag
      \llbracket \min(a,b)\rrbracket^A=\mathsf{MultiBA}(\llbracket\varphi\rrbracket^B,\llbracket a\rrbracket^A)+\mathsf{MultiBA}(\llbracket\varphi'\rrbracket^B,\llbracket b\rrbracket^A),
    \end{equation} 
    where $\llbracket\varphi'\rrbracket^B=\llbracket\neg\varphi\rrbracket^B$. 
    With this as a basis, we are able to compute the minimum attribute sum in the secret sharing domain, as well as obliviously fetch the corresponding skyline tuple $\mathbf{t}_{\star}$ from $\llbracket\mathbf{T}\rrbracket^A$ and the corresponding tuple $\mathbf{p}_{\star}$ from $\llbracket\mathbf{P}\rrbracket^A$.
    In particular, when securely switching two attribute sums based on the secret-shared comparison result, we can perform secure switching of the two associated secret-shared tuples from $\llbracket\mathbf{T}\rrbracket^A$ and $\llbracket\mathbf{P}\rrbracket^A$ as well.
    The details of secure skyline fetching is presented in Algorithm \ref{alg:bFind}.
    Note that there may be more than one tuple whose attribute sum is equal to the smallest value in $\mathbf{s}$, but only
    one of them needs to be fetched in the current round, because the remainders are the skyline tuples to be 
    processed in the subsequent rounds.
    Therefore, {\main} lets {\csa} obliviously choose the first one that has the minimum attribute sum by performing
    \begin{equation}
        \llbracket\varphi\rrbracket^B=\mathsf{SecExt}(\llbracket\mathbf{s}[i]\rrbracket^A,\llbracket sMin\rrbracket^A),
    \end{equation}
    where $\varphi=1$ if and only if $\mathbf{s}[i]<sMin$, which keeps $\mathbf{t}_{\star}$ and $\mathbf{p}_{\star}$
    unchanged when $\mathbf{s}[i]=sMin$.

    Note that when implementing Algorithm \ref{alg:bFind}, we can use the trick of divide-and-conquer \cite{cole1988parallel} to boost the performance during secure minimum computation.
    For example, the minimum in a vector $\llbracket\mathbf{v}\rrbracket^A$ of four elements can be calculated by: $\min(\min(\llbracket\mathbf{v}[1]\rrbracket^A,\llbracket\mathbf{v}[2]\rrbracket^A),\min(\llbracket\mathbf{v}[3]\rrbracket^A,\llbracket\mathbf{v}[4]\rrbracket^A))$, where $\min$ $(\llbracket\mathbf{v}[1]\rrbracket^A,\llbracket\mathbf{v}[2]\rrbracket^A)$ and $\min(\llbracket\mathbf{v}[3]\rrbracket^A,\llbracket\mathbf{v}[4]\rrbracket^A)$ can be calculated in parallel, saving communication rounds.

    \subsection{Secure Skyline and Dominated Tuples \revise{Filtering}}
    \label{subsec:elim}

     So far we have introduced how {\csa} obliviously fetch the encrypted skyline tuple $\llbracket\mathbf{t}_{\star}\rrbracket^A$  from the encrypted mapped database $\llbracket\mathbf{T}\rrbracket^A$.
     Then we should consider how to allow {\csa} to obliviously filter out $\llbracket\mathbf{t}_{\star}\rrbracket^A$ and the tuples  dominated by $\llbracket\mathbf{t}_{\star}\rrbracket^A$ from $\llbracket\mathbf{T}\rrbracket^A$ without knowing which tuples they are, i.e., hiding the access pattern and the dominance relationships.
     Therefore, we devise a component $\mathsf{secFilt}$ (as given in Algorithm \ref{alg:bEli}) for secure skyline and dominated tuples \revise{ filtering}. %, which inputs the encrypted skyline tuple $\llbracket\mathbf{t}_{\star}\rrbracket^A$, database $\llbracket\mathbf{T}\rrbracket^A$, minimum sum of attributes $\llbracket  sMin\rrbracket^A$ and sum of attributes $\llbracket\mathbf{s}^{(k)}\rrbracket^A$, and outputs the new encrypted sum of attributes $\llbracket\mathbf{s}^{(k+1)}\rrbracket^A$ after elimination.

 \begin{algorithm}[t!]
        \caption{Secure Skyline and Dominated Tuples \revise{ Filtering} \revise{ $\mathsf{secFilt}$}}
        \label{alg:bEli}
        \begin{algorithmic}[1]

            \REQUIRE The encrypted database $\llbracket\mathbf{T}\rrbracket^A$, skyline tuple $\llbracket\mathbf{t}_{\star}\rrbracket^A$, minimum sum of attributes $\llbracket  sMin\rrbracket^A$ and sum of attributes $\llbracket\mathbf{s}^{(k)}\rrbracket^A$.
            \ENSURE The new encrypted sum of attributes $\llbracket\mathbf{s}^{(k+1)}\rrbracket^A$.
            
            \STATE Set $\llbracket\mathsf{flag}\rrbracket^B=\llbracket0\rrbracket^B$.
            
            \FOR{$i=1$ to $n$}

               \STATE $\llbracket \sigma_{i}\rrbracket^B=\neg\mathsf{SecExt}(\llbracket sMin\rrbracket^A,\llbracket\mathbf{s}^{(k)}[i]\rrbracket^A)$.\label{alg4:mark_sky_beg}

                \STATE $\llbracket\mathsf{isFirstSky}_{i}\rrbracket^B=\llbracket \sigma_{i}\rrbracket^B\otimes \llbracket\neg\mathsf{flag}\rrbracket^B$.

               \STATE $\llbracket\mathsf{flag}\rrbracket^B=\llbracket\mathsf{flag}\rrbracket^B\oplus \llbracket\mathsf{isFirstSky}_{i}\rrbracket^B$.\label{line:EliAMB1}\label{alg4:mark_sky_end}
               
               \FOR{$j=1$ to $m$}\label{alg4:select_begin}
               \STATE $\llbracket\delta_{i,j}\rrbracket^B=\mathsf{SecExt}(\llbracket\mathbf{t}_i[j]\rrbracket^A,\llbracket\mathbf{t}_{\star}[j]\rrbracket^A)$.\label{line:EliMSB1}
               
               \ENDFOR
               
               \STATE $\llbracket \hat{\delta}_{i}\rrbracket^B=\llbracket\neg\delta_{i,1}\rrbracket^B\otimes\cdots\otimes\llbracket\neg\delta_{i,m}\rrbracket^B$. \label{alg4:select_end}
               
               \STATE$\llbracket\mathsf{isDomi}_{i}\rrbracket^B=\llbracket \hat{\delta}_{i}\rrbracket^B\otimes \llbracket \neg\sigma_{i}\rrbracket^B$.

                \STATE $\llbracket\Phi_{i}\rrbracket^B=\llbracket\mathsf{isFirstSky}_{i}\rrbracket^B \oplus\llbracket\mathsf{isDomi}_{i}\rrbracket^B$. \label{alg4:elim}

               \STATE $\llbracket\mathbf{s}^{(k+1)}[i]\rrbracket^A= \mathsf{MultiBA}(\llbracket\Phi_{i}\rrbracket^B,\llbracket\mathsf{vMAX}\rrbracket^A)+$\\$ \mathsf{MultiBA}(\llbracket \neg\Phi_{i}\rrbracket^B,\llbracket\mathbf{s}^{(k)}[i]\rrbracket^A)$.\label{line:EliAMB2}
               
            \ENDFOR
            \RETURN $\llbracket\mathbf{s}^{(k+1)}[i]\rrbracket^A$.
        \end{algorithmic}
    \end{algorithm}

    \noindent\textbf{Challenges.} There are two challenges to be tackled in $\mathsf{secFilt}$: 1) how to allow {\csa} to obliviously locate the skyline tuple and dominated tuples in $\llbracket\mathbf{T}\rrbracket^A$? 2) how to allow {\csa} to obliviously filter out these tuples?

    \noindent\textbf{Addressing the first challenge.}
    {\main} first defines two encrypted (binary) labels $\llbracket\mathsf{isFirstSky}_i\rrbracket^B$ and $\llbracket\mathsf{isDomi}_i\rrbracket^B$ for each tuple $\llbracket\mathbf{t}_{i}\rrbracket^A\in\llbracket\mathbf{T}\rrbracket^A$,
    where $\mathsf{isFirstSky}_i=1$ indicates that $\mathbf{t}_i$ is the 
    skyline tuple in the current round and $\mathsf{isDomi}_i=1$ indicates that $\mathbf{t}_i$
    is a tuple dominated by the skyline tuple.
    Then {\csa} can obliviously mark whether tuple
    $\mathbf{t}_{i}$ needs to be filtered out by calculating
    \begin{equation}\label{eq:Phi_i}
        \llbracket\Phi_{i}\rrbracket^B=\llbracket\mathsf{isFirstSky}_{i}\rrbracket^B\ \oplus \llbracket\mathsf{isDomi}_{i}\rrbracket^B.
    \end{equation}
    It is noted that since $\mathsf{isFirstSky}_i$ and $\mathsf{isDomi}_i$ cannot both be equal to 1, $\Phi_{i}=0$ indicates that both $\mathsf{isFirstSky}_i$ and $\mathsf{isDomi}_i$ are equal to 0 and   $\mathbf{t}_{i}$ does not need to be filtered out, and  $\Phi_{i}=1$ indicates that $\mathsf{isFirstSky}_i$ or $\mathsf{isDomi}_i$ is equal to 1 and  $\mathbf{t}_{i}$ needs to be filtered out.
    Next, we introduce how {\csa} obliviously evaluate $\llbracket\mathsf{isFirstSky}_i\rrbracket^B$ and $\llbracket\mathsf{isDomi}_i\rrbracket^B$ for each tuple $\llbracket\mathbf{t}_{i}\rrbracket^A\in\llbracket\mathbf{T}\rrbracket^A$.
    
    We first introduce how {\csa} obliviously evaluate $\llbracket\mathsf{isFirstSky}_i\rrbracket^B$ for each tuple $\llbracket\mathbf{t}_{i}\rrbracket^A\in\llbracket\mathbf{T}\rrbracket^A$, i.e. whether $\mathbf{t}_i$ is the 
    skyline tuple.
    Recall that in Algorithm \ref{alg:bFind}, the skyline tuple has the minimum attribute sum $sMin$. so {\main} lets {\csa} obliviously evaluate whether $\mathbf{t}_i$'s attribute
    sum (i.e., $\mathbf{s}^{(k)}[i]$ in the current round $k$) is equal to $sMin$.
    Specifically, {\main} first lets {\csa} securely compare $\llbracket sMin\rrbracket^A$ and $\llbracket\mathbf{s}^{(k)}[i]\rrbracket^A$ by 
    \begin{equation}\label{eq:sigma}
        \llbracket \sigma_{i}\rrbracket^B=\neg\mathsf{SecExt}(\llbracket sMin\rrbracket^A,\llbracket\mathbf{s}^{(k)}[i]\rrbracket^A),
    \end{equation}
    where $\sigma_{i}=1$ indicates $sMin\geq \mathbf{s}^{(k)}[i]$.
    Note that $sMin$ is the minimum value in $\mathbf{s}^{(k)}$, and thus $\sigma_{i}=1$
    means $ \mathbf{s}^{(k)}[i]=sMin $.
    However, since there may be more than one value in $\mathbf{s}^{(k)}$ that is equal to 
    $sMin$, $ \mathbf{s}^{(k)}[i]= sMin $ indicates that $\mathbf{t}_{i}$ \textit{may} 
    be the skyline tuple $\mathbf{t}_{\star}$.
    Recall that in Algorithm \ref{alg:bFind}, {\main} lets {\csa} obliviously fetch the \textit{first} tuple $\llbracket\mathbf{t}_{i}\rrbracket^A$ whose $\llbracket\mathbf{s}[i]\rrbracket^A$ is minimum in  $\llbracket\mathbf{s}\rrbracket^A$ (i.e., $\llbracket sMin \rrbracket^A$) as the skyline tuple $\llbracket\mathbf{t}_{\star}\rrbracket^A$.
    Therefore, {\main} provides a delicate security design to allow {\csa} to only set $\llbracket\mathsf{isFirstSky}_{i}\rrbracket^B=\llbracket 1\rrbracket^B$ for the \textit{first} tuple $\llbracket\mathbf{t}_{i}\rrbracket^A$ that satisfies $\llbracket \mathbf{s}^{(k)}[i]\rrbracket^A=\llbracket sMin \rrbracket^A$ as follows:
    \begin{align}
    &\llbracket\mathsf{isFirstSky}_{i}\rrbracket^B=\llbracket \sigma_{i}\rrbracket^B\otimes \llbracket\neg\mathsf{flag} \rrbracket^B,\label{eq:first_sky_set}\\ 
    &\llbracket\mathsf{flag}\rrbracket^B=\llbracket\mathsf{flag}\rrbracket^B\oplus \llbracket\mathsf{isFirstSky}_{i}\rrbracket^B,\label{eq:flag_set}
    \end{align}
    where $\llbracket\mathsf{flag}\rrbracket^B$ is an auxiliary variable and set as $\llbracket0\rrbracket^B$ at the beginning of the current round.
    $\mathsf{isFirstSky}_{i}=1$ indicates that tuple $\mathbf{t}_{i}$ is the required skyline tuple $\mathbf{t}_{\star}$.

    The correctness is analyzed as follows. 
    $\llbracket\mathsf{flag}\rrbracket^B=\llbracket0\rrbracket^B$ at the beginning.
    When the first $\llbracket \mathbf{s}^{(k)}[i]\rrbracket^A=\llbracket sMin \rrbracket^A$ appears, {\csa} obliviously set $\llbracket \sigma_{i}\rrbracket^B=\llbracket 1\rrbracket^B$.
    Since $\llbracket\neg\mathsf{flag}\rrbracket^B=\llbracket1\rrbracket^B$, {\csa} obliviously set $\llbracket\mathsf{isFirstSky}_{i}\rrbracket^B=\llbracket 1\rrbracket^B$ (i.e., Eq. \ref{eq:first_sky_set}), which marks the required skyline tuple $\mathbf{t}_{\star}$.
    After that,  {\csa} obliviously set $\llbracket\mathsf{flag}\rrbracket^B=\llbracket1\rrbracket^B$, i.e., Eq. \ref{eq:flag_set}, and $\llbracket\mathsf{flag}\rrbracket^B$ remains equal to $\llbracket1\rrbracket^B$ in the following loops.
    $\llbracket\mathsf{flag} \rrbracket^B=\llbracket1\rrbracket^B$ prevents {\csa} from setting  $\llbracket\mathsf{isFirstSky}_{i}\rrbracket^B=\llbracket 1\rrbracket^B$ for other tuples because $\llbracket \neg\mathsf{flag} \rrbracket^B=\llbracket0\rrbracket^B$ in the following loops.

    We then introduce how {\csa} obliviously evaluate $\llbracket\mathsf{isDomi}_i\rrbracket^B$ for each tuple $\llbracket\mathbf{t}_{i}\rrbracket^A\in\llbracket\mathbf{T}\rrbracket^A$, i.e. whether $\mathbf{t}_i$ is a dominated tuple.
   	\revise{According to the definition of dominance (i.e., Definition \ref{def:skyline}), given two tuples $\mathbf{a}$ and $\mathbf{b}$, we say $\mathbf{a}$ dominates $\mathbf{b}$
    if $\forall j$, $\mathbf{a}[j]\leq\mathbf{b}[j]$ and $\exists j$, $\mathbf{a}[j]<\mathbf{b}[j]$.
    Therefore, if $\forall j$, $\mathbf{a}[j]\leq\mathbf{b}[j]$, we have either $\mathbf{a}$ dominates
    $\mathbf{b}$ or $\mathbf{a}$ is identical to $\mathbf{b}$.
    Therefore, {\main} defines an encrypted (binary) label $\llbracket\hat{\delta}_{i}\rrbracket^B$ to mark the above dominance relationship.
    {\csa} obliviously evaluate  $\llbracket\hat{\delta}_{i}\rrbracket^B$ for each tuple $\llbracket\mathbf{t}_{i}\rrbracket^A\in\llbracket\mathbf{T}\rrbracket^A$ by first comparing each attribute of the skyline tuple $\llbracket\mathbf{t}_{\star}\rrbracket^A$ and $\llbracket\mathbf{t}_{i}\rrbracket^A$:
    \begin{equation}\notag
    \llbracket\delta_{i,j}\rrbracket^B=\mathsf{SecExt}(\llbracket\mathbf{t}_i[j]\rrbracket^A,\llbracket\mathbf{t}_{\star}[j]\rrbracket^A),j\in[1,m],
    \end{equation}
    and then aggregating the comparison results to $\llbracket\hat{\delta}_{i}\rrbracket^B$ by 
    \begin{equation}\label{eq:hat_delta}
    \llbracket \hat{\delta}_{i}\rrbracket^B=\llbracket\neg\delta_{i,1}\rrbracket^B\otimes\cdots\otimes\llbracket\neg\delta_{i,m}\rrbracket^B, 
    \end{equation} 
    where $\hat{\delta}_{i}=1$ if and only if $\forall\delta_{i,j}=0,j\in[1,m]$, i.e., $\forall j\in[1,m]$, $\mathbf{t}_{\star}[j]\leq\mathbf{t}_{i}[j]$.
    Therefore, $\hat{\delta}_{i}=1$ if $\mathbf{t}_{i}$ is a dominated tuple or $\mathbf{t}_{i}=\mathbf{t}_{\star}$, and otherwise $\hat{\delta}_{i}=0$.}
    The above process is described at lines \ref{alg4:select_begin}-\ref{alg4:select_end} of Algorithm \ref{alg:bEli}.
    Subsequently, {\csa} securely evaluate $\llbracket\mathsf{isDomi}_{i}\rrbracket^B$ for $\llbracket\mathbf{t}_{i}\rrbracket^A$ by
    \begin{equation}\notag
        \llbracket\mathsf{isDomi}_{i}\rrbracket^B=\llbracket \hat{\delta}_{i}\rrbracket^B\otimes \llbracket \neg\sigma_{i}\rrbracket^B,
    \end{equation} 
    where $\mathsf{isDomi}_{i}=1$ indicates that tuple $\mathbf{t}_{i}$ is a dominated tuple, and $\llbracket \hat{\delta}_{i}\rrbracket^B$ and $\llbracket\sigma_{i}\rrbracket^B$ are calculated by Eq. \ref{eq:hat_delta} and Eq. \ref{eq:sigma}, respectively. 
    We give the correctness analysis as follows.
    Firstly, if and only if  both $\hat{\delta}_{i}=1$ and   $ \neg\sigma_{i}=1$, we have $\mathsf{isDomi}_{i}=1$.
    $\hat{\delta}_{i}=1$ indicates that $\mathbf{t}_{i}$ is a dominated tuple or $\mathbf{t}_{i}=\mathbf{t}_{\star}$, and $ \neg\sigma_{i}=1$ rules out the possibility of $\mathbf{t}_{i}=\mathbf{t}_{\star}$.
    Therefore, if and only if the tuple $\mathbf{t}_{i}$ is a dominated tuple, we have $\mathsf{isDomi}_{i}=1$.

    So far, {\csa} have obliviously marked the skyline tuple by $\mathsf{isFirstSky}_{i}=1$
    and the dominated tuples by $\mathsf{isDomi}_{i}=1$.
    Therefore, {\csa} can securely evaluate $\llbracket\Phi_{i}\rrbracket^B$ for each tuple $\llbracket\mathbf{t}_{i}\rrbracket^A\in\llbracket\mathbf{T}\rrbracket^A$ by Eq. \ref{eq:Phi_i} to obliviously mark whether tuple $\mathbf{t}_{i}$ needs to be filtered out.

	\noindent\textbf{Addressing the second challenge.} We should then consider how to tackle the second challenge, namely, how to allow {\csa} to obliviously filter out tuples which have $\Phi_{i}=1$ underlying the $\llbracket \Phi_{i}\rrbracket^B$ (recall Eq. \ref{eq:Phi_i}).
        A naive method is to let {\csa} directly open each tuple's $ \Phi_{i}$.
        However, such simple method will leak which tuple is the skyline tuple and the dominance relationships, which easily violates the security requirement for access pattern protection.

        Instead, {\main} achieves oblivious tuple filtering via a different strategy.
        Specifically, {\main} lets {\csa} obliviously set $\mathbf{s}^{(k+1)}[i]$ to a pre-set system-wide maximum value $\mathsf{vMAX}$ to mark for filtering if $\Phi_{i}= 1$, and obliviously keep $\mathbf{s}^{(k+1)}[i]$ unchanged if $\Phi_{i}= 0$.
        Formally, {\csa} perform the following:
        \begin{align}\notag
        	\llbracket\mathbf{s}^{(k+1)}[i]\rrbracket^A= \mathsf{MultiBA}(\llbracket\Phi_{i}\rrbracket^B,\llbracket\mathsf{vMAX}\rrbracket^A)+\\ \mathsf{MultiBA}(\llbracket \neg\Phi_{i}\rrbracket^B,\llbracket\mathbf{s}^{(k)}[i]\rrbracket^A).\notag
        \end{align}
        The secret sharing of $\mathsf{vMAX}$ can be prepared by {\csa} offline.
        Note that $\mathbf{s}^{(k+1)}[i]=\mathsf{vMAX}$ will prevent {\csa} from selecting $\mathbf{t}_{i}$ (or $\mathbf{p}_{i}$) as the skyline tuple in the following rounds.
        Therefore, the filtering strategy will not degrade the skyline query accuracy.

    \subsection{Putting Things Together}
    
   \label{subsec:secure_main}
    
    In this section, we introduce how to synthesize the above three secure components to enable secure skyline query processing over encrypted cloud databases in {\main}.
    We first encapsulate them as follows: 
    \begin{itemize}
    	\item Secure database mapping $\llbracket\mathbf{T}\rrbracket^A=\mathsf{secMap}(\llbracket\mathbf{P}\rrbracket^A,$ $\llbracket\mathbf{q}\rrbracket^A)$, which inputs the encrypted original database $\llbracket\mathbf{P}\rrbracket^A$ and skyline query request $\llbracket\mathbf{q}\rrbracket^A$, and then outputs the initial encrypted mapped database $\llbracket\mathbf{T}\rrbracket^A$.

    	\item  Secure skyline fetching $(\llbracket  sMin\rrbracket^A$, $\llbracket\mathbf{t}_{\star}\rrbracket^A$, $\llbracket\mathbf{p}_{\star}\rrbracket^A)=\mathsf{secFetch}(\llbracket\mathbf{s}^{(k)}\rrbracket^A, \llbracket\mathbf{T}\rrbracket^A, \llbracket\mathbf{P}\rrbracket^A)$, which inputs the encrypted attribute sum vector $\llbracket\mathbf{s}^{(k)}\rrbracket^A$, mapped database $\llbracket\mathbf{T}\rrbracket^A$, and the original database $\llbracket\mathbf{P}\rrbracket^A$ output from the previous round, and then outputs the encrypted minimum attribute sum $\llbracket  sMin\rrbracket^A\in \llbracket\mathbf{s}^{(k)}\rrbracket^A$, the skyline tuples $\llbracket\mathbf{t}_{\star}\rrbracket^A$ and $\llbracket\mathbf{p}_{\star}\rrbracket^A$.

    	\item Secure skyline and dominated tuples filtering $\llbracket\mathbf{s}^{(k+1)}\rrbracket^A= \mathsf{secFilt}(\llbracket\mathbf{T}\rrbracket^A,$ $ \llbracket\mathbf{t}_{\star}\rrbracket^A,\llbracket  sMin\rrbracket^A,\llbracket\mathbf{s}^{(k)}\rrbracket^A)$, which inputs the encrypted mapped database $\llbracket\mathbf{T}\rrbracket^A$, minimum sum $\llbracket  sMin\rrbracket^A$, and attribute sum vector $\llbracket\mathbf{s}^{(k)}\rrbracket^A$, and then outputs the updated attribute sum vector $\llbracket\mathbf{s}^{(k+1)}\rrbracket^A$.

    \end{itemize}
     Algorithm \ref{alg:overview} gives the complete construction for secure skyline query processing in {\main}, which is the secure instantiation of Algorithm \ref{alg:skyline_plaintext} and relies on the coordination of the above three secure components.
    The only challenge in the design of Algorithm \ref{alg:overview} is how to allow {\csa} to decide whether to terminate the secure search process without leaking other information.
    Our solution is to let {\csa} first securely compare $\llbracket  sMin\rrbracket^A$ and $ \mathsf{vMAX}$:
   \begin{equation}\notag
    \llbracket \mathsf{isStop}\rrbracket^B=\neg\mathsf{SecExt}(\llbracket sMin\rrbracket^A,\llbracket\mathsf{vMAX}\rrbracket^A),
   \end{equation}
    and then open the flag $\mathsf{isStop}$, where $\mathsf{isStop}=1$ indicates that the smallest element in $\llbracket\mathbf{s}^{(k)}\rrbracket^A$ is $\mathsf{vMAX}$.
    Therefore, {\csa} can know that all elements in $\llbracket\mathbf{s}^{(k)}\rrbracket^A$ have been filtered out, and then terminate the search process.

        \begin{algorithm}[t]
       	\caption{The Complete Construction for Secure Skyline Query Processing in {\main}}
       	\label{alg:overview}
       	\begin{algorithmic}[1]
       		\REQUIRE The encrypted database $\llbracket{\mathbf{P}}\rrbracket^A$ and skyline query $\llbracket\mathbf{q}\rrbracket^A$.
       		
       		\ENSURE The encrypted resulting set of skyline tuples $\llbracket\mathcal{SK}_{\mathbf{q}}\rrbracket^A$.
       		
       		\STATE Initialization: $\llbracket\mathcal{SK}_{\mathbf{q}}\rrbracket^A=\emptyset, \llbracket\mathbf{s}\rrbracket^A=\llbracket\mathbf{0}\rrbracket^A,  $ \revise{$\mathsf{isStop}=0$}.
       		
       		\STATE $\llbracket\mathbf{T}\rrbracket^A= \mathsf{secMap}(\llbracket\mathbf{P}\rrbracket^A,\llbracket\mathbf{q}\rrbracket^A)$.\label{line:oMapping}
       		\STATE $\llbracket\mathbf{s}^{(0)}[i]\rrbracket^A=\sum_{j=1}^{m}\llbracket\mathbf{t}_i[j]\rrbracket^A$, for $i\in[1,n]$. // $n$ is the number of tuples; $m$ is the dimension of tuples.\label{line:oSum}
       		\STATE $k=0$.
       		\WHILE{$\neg\mathsf{isStop}$}
       		\STATE $(\llbracket  sMin\rrbracket^A,\llbracket\mathbf{t}_{\star}\rrbracket^A,\llbracket\mathbf{p}_{\star}\rrbracket^A)=$\\ $\mathsf{secFetch}(\llbracket\mathbf{s}^{(k)}\rrbracket^A,\llbracket\mathbf{T}\rrbracket^A, \llbracket\mathbf{P}\rrbracket^A)$.\label{line:oFinding}
       		\STATE $\llbracket \mathsf{isStop}\rrbracket^B=\neg\mathsf{SecExt}(\llbracket  sMin\rrbracket^A,\llbracket\mathsf{vMAX}\rrbracket^A)$.\label{line:OverviewMSB}
       		\STATE {\csa} open the flag $\mathsf{isStop}$ to decide whether to stop the process.\label{line:OverviewComm}
%       		\IF{$b=0$}
%       		\STATE Break.
%       		\ENDIF
       		\STATE $\llbracket\mathcal{SK}_{\mathbf{q}}\rrbracket^A.append(\llbracket\mathbf{p}_{\star}\rrbracket^A)$.
       		
       		\STATE $\llbracket\mathbf{s}^{(k+1)}\rrbracket^A= \mathsf{secFilt}(\llbracket\mathbf{T}\rrbracket^A, \llbracket\mathbf{t}_{\star}\rrbracket^A,\llbracket  sMin\rrbracket^A,\llbracket\mathbf{s}^{(k)}\rrbracket^A)$.\label{line:oEli}
       		\STATE $k++$.
       		
       		\ENDWHILE
       		%\STATE $C_1$ sends $\langle\mathbf{Sky}\rangle^A_1$ to Client and $C_2$ sends $\langle\mathbf{Sky}\rangle^A_2$ to Client.
       		%\STATE \textbf{Client:}
       		%\STATE Reveal the skyline query result: $\mathbf{Sky}=\langle\mathbf{Sky}\rangle^A_1+\langle\mathbf{Sky}\rangle^A_2$.
       		  \RETURN  $\llbracket\mathcal{SK}_{\mathbf{q}}\rrbracket^A$.
       	\end{algorithmic}
       \end{algorithm}

    \subsection{Complexity Analysis}
     \label{subsec:complexity_analysis}
    % We analyze the theoretical performance complexity of {\main}.
    %
    It is noted that the performance of {\main} is dominated by three main components: 1) secure database mapping ($\mathsf{secMap}$), 2) secure skyline fetching ($\mathsf{secFetch}$), and 3) secure skyline and dominated tuples filtering ($\mathsf{secFilt}$).
    In practice, the cost of these components is dominated by the secure MSB extraction operation $\mathsf{SecExt(\cdot)}$.
    Therefore, we first separately analyze their complexities by counting $\mathsf{SecExt(\cdot)}$ during their execution.
    For Algorithm \ref{alg:Mapping}, the cost of $\mathsf{secMap}$ is dominated by $n\cdot m$ secure MSB extraction operations $\mathsf{SecExt(\cdot)}$.
    For Algorithm \ref{alg:bFind}, $\mathsf{secFetch}$ requires $n-1$ secure MSB extraction operations $\mathsf{SecExt(\cdot)}$.
    For Algorithm \ref{alg:bEli}, $\mathsf{secFilt}$ requires $n\cdot m+n$ secure MSB extraction operations $\mathsf{SecExt(\cdot)}$.

    Then we analyze the overall complexity of invoking $\mathsf{SecExt(\cdot)}$ in the complete protocol shown in Algorithm \ref{alg:overview}. 
    %
    % Algorithm \ref{alg:overview} describes the overall design of {\main}.
    %
    Note that it is dominated by the $\mathsf{While}$ loop, which terminates when all tuples in $\llbracket T\rrbracket^A$ are filtered out, as indicated by the Boolean flag $\mathsf{isStop}$.
    %  {\csa} find all skyline tuples, as indicated by the Boolean flag $\mathsf{isStop}$.
    %
    The number of loops is $k+1$ where $k$ is the number of skyline tuples returned for the skyline query, i.e., the size of $\llbracket\mathcal{SK}_{\mathbf{q}}\rrbracket^A$.
    % which is the number of skyline tuples returned for the skyline query, i.e., the size of $\llbracket\mathcal{SK}_{\mathbf{q}}\rrbracket^A$.
    %
    %Because these $k$ skyline tuples are fetched in the first $k$ rounds, and it requires an extra fetching to set $\mathsf{isStop}=1$.
    %
    In addition, the computation of $\llbracket \mathsf{isStop}\rrbracket^B$ in each loop requires a secure MSB extraction operation $\mathsf{SecExt(\cdot)}$ and the $(k+1)$-th loop early stops at line 8.
    Therefore, we can conclude that the overall execution of our protocol requires $n\cdot m+k\cdot n\cdot(2+m)+n$ secure MSB extraction operations  $\mathsf{SecExt(\cdot)}$.

    \section{Security Analysis}
    \label{sec:security}
    We now analyze the security of {\main}. \revise{Our analysis follows the standard ideal/real world paradigm \cite{lindell2017how}}.
    We first define the ideal functionality $\mathcal{F}$ for the secure skyline query processing:
    \begin{itemize}
        \item \textbf{Input.} The data owner provides to $\mathcal{F}$  the database $\mathbf{P}$ and a client submits a query $\mathbf{q}$. 
        \item \textbf{Computation.} After receiving $\mathbf{P}$ and $\mathbf{q}$, $\mathcal{F}$ retrieves the skyline tuples $\mathcal{SK}_{\mathbf{q}}$ of $\mathbf{P}$ with respect to $q$.
        \item \textbf{Output.} $\mathcal{F}$ returns $\mathcal{SK}_{\mathbf{q}}$ to the client.
    \end{itemize}
    Let $\prod$ represent a protocol for secure skyline query processing that realizes the ideal functionality $\mathcal{F}$. The security of $\prod$ is formally defined as follows:

    \begin{definition}
        \label{def:secAna}
        Let $\mathcal{A}$ be an adversary who has the view of a corrupted server during the execution of $\prod$.
        Let $\mathsf{View}^{\mathsf{Real}}_{\prod(\mathcal{A})}$ denote $\mathcal{A}$'s view in the real world.
        %
        % In the ideal world, a simulator $\mathcal{S}$ can generate the simulated view $View^{\mathsf{Ideal}}_{\mathcal{S}}$.
        %
        We say that $\prod$ is secure in the semi-honest and non-colluding setting, if for $\forall$ PPT adversary, $\exists$ a PPT simulator $\mathcal{S}$ s.t. $\mathsf{View}_{\prod(\mathcal{A})}^{\mathsf{Real}}\mathop \approx \limits\mathsf{View}_{\mathcal{S}}^{\mathsf{Ideal}}$. That is, the simulator $\mathcal{S}$ can simulate a view for the adversary, which is indistinguishable from its view in the real-world. 
    \end{definition}

    \begin{theorem}
        \label{th:secSkyline}
        In the semi-honest and non-colluding adversary model, {\main} can securely realize the ideal functionality $\mathcal{F}$ according to \textit{Definition} \ref{def:secAna}.
    \end{theorem}

    \begin{proof}
    Recall that in the framework of {\main}, i.e., Algorithm \ref{alg:overview}, which consists of several components: 1) secure database mapping ($\mathsf{secMap}$); 2) secure skyline fetching ($\mathsf{secFetch}$); 3) secure skyline and dominated tuples \revise{ filtering} (\revise{ $\mathsf{secFilt}$}).
     Since each of them is invoked in order as per the processing pipeline and their inputs and outputs are secret shares, we can conclude that {\main} is secure if the simulator for each component exists \cite{canetti2000security,katz2005handling,curran2019procsa}.
    We use $Sim^{C_{i}}_{\mathtt{X}}$ to represent the simulator which generates $C_{i}$'s view in the execution of component $\mathtt{X}$ on corresponding input and output.
   It is noted that the roles of $C_1$ and $C_2$ in these components are symmetric. So it suffices to analyze the existence of simulators for $C_1$.

    \begin{itemize}
        \item $Sim^{C_1}_{\mathsf{secMap}}$. %\textbf{Simulators for $C_1$ in secMap.}
        It is noted that $\mathsf{secMap}$ (i.e.,  Algorithm \ref{alg:Mapping}) consists of three meta operations, i.e., secure MSB extraction (line \ref{line:MapMSB}), secure bit flipping (line \ref{line:MapNOT}), and secure multiplication between a binary secret-shared value and an arithmetic secret-shared value (line \ref{line:MapAMB}). 
        Since these operations are invoked in turn and their inputs are secret shares, we analyze the existence of their simulators in turn.
        Since the secure MSB extraction consists of basic binary secret sharing operations (i.e., AND $\otimes$ and XOR $\oplus$), its simulator clearly exists.
        Note that the secure bit flipping operation only requires local computation and $C_1$ receives nothing during its execution.
        Therefore, its simulator clearly exists.
        We then analyze the existence of the simulator for the secure multiplication between a binary secret-shared value $\llbracket x\rrbracket^B$ and an arithmetic secret-shared value $\llbracket y\rrbracket^A$.
        %
        %Recall the Algorithm \ref{alg:Mapping}, except the $msb$ function, $C_1$ only receives shares in the multiplication between arithmetic secret-shared number and binary secret-shared number (AMB).
        %
        %Then we need to proof the simulator for AMB exists.
        %
        %Assuming that binary sharing $\llbracketmsb(x)\rrbracket^B$ multiply arithmetic sharing $\llbrackety\rrbracket^A$.
        %
        %The multiplication consists of a round of symmetrical communication, and two cloud servers $C_{\{1, 2\}}$ play the symmetrical role during this multiplication.
        %
        %So, we assume that $C_1$ acts as a sender and $C_2$ acts as a receiver.
        %
        %
        It is noted that we only need to analyze the case where $C_{1}$ acts as the receiver, because in the case where $C_{1}$ acts as the sender, $C_{1}$ receives nothing.
        At the beginning of the operation, $C_{1}$ has $\langle x\rangle^B_1$ and $\langle y\rangle^A_1$, and later receives two messages $m_\mu:=(\mu\oplus\langle x\rangle^B_2)\cdot\langle y\rangle^A_2-r_2,\mu\in\{0,1\}$ from $C_{2}$.
        Therefore, we need to prove that the messages are uniformly random in the view of $C_1$.
        Note that the random value $r_2$ generated by $C_2$ is uniformly random in the view of $C_1$.
        This implies that $m_{\{1,2\}}$ are also uniformly random in $C_1$'s view since $r_2$ is independent of other values used in the generation of $m_{\{1,2\}}$ \cite{araki2016high}.
        %\begin{itemize}
            %\item $Sim^{C_1}_{AMB}$: As a sender, $C_1$ only send but receives nothing during the execution in real.
            %
            %Thus, the view of simulator in ideal is identical to the view in real.
            %\item $Sim^{C_2}_{AMB}$: As a receiver, $C_2$ has $\langle msb(x)\rangle^B_2$ and $\langle y\rangle^A_2$ at the beginning, and receives the messages $m_b:=(b\oplus\langle msb(x)\rangle^B_1)\times\langle y\rangle^A_1-r,b\in\{0,1\}$.
            %
           % Then we need to proof the messages are uniformly random in the view of $C_2$.
            %
            %Notice that, the random value $r$ generated by $C_1$ is uniformly random in the view of $C_2$.
            %
            %And in the computations of messages, $r$ is independent of other values, based which we can say that messages $m_{\{1,2\}}$ are uniformly random in the view of $C_2$.
            %
            %Therefore, there exists a simulator that can simulated messages $m_{\{1,2\}}$, which are computationally indistinguishable with the real $m_{\{1,2\}}$ received from $C_1$ in the view of $C_2$.
        %\end{itemize}
        %
        Therefore, the simulator $Sim^{C_1}_{\mathsf{secMap}}$ exists.
        
        \item $Sim^{C_1}_{\mathsf{secFetch}}$. %\textbf{Simulators for $C_1$ in secFetch.}
        It is noted that $\mathsf{secFetch}$ (i.e., Algorithm \ref{alg:bFind}) consists of secure MSB extraction, secure bit flipping, secure multiplication between a binary secret-shared value and an arithmetic secret-shared value, and basic secret sharing operations.
        Meanwhile, these operations are invoked in turn and their inputs are secret shares.
        Therefore, based on the above analysis, the simulator $Sim^{C_1}_{\mathsf{secFetch}}$ exists.

        %Recall the Algorithm \ref{alg:bFind}, $C_1$ receives the intermediate shares of $msb$ function and AMB.
        %
        %According to the proof above, the simulator $Sim^{C_1}_{secFetch}$ exists.
        \item $Sim^{C_1}_{\mathsf{secFilt}}$. %\textbf{Simulators for $C_1$ in secFilt.}
        Similarly, $Sim^{C_1}_{\mathsf{secFilt}}$ exists, since the meta operations of \revise{ $\mathsf{secFilt}$} (i.e., Algorithm \ref{alg:bEli}) are same as $\mathsf{secMap}$ and $\mathsf{secFetch}$, and they as are invoked in turn and their inputs are secret shares.
    \end{itemize}
    The proof of Theorem \ref{th:secSkyline} is completed.
\end{proof}

      We now explicitly analyze why {\main} can hide search access patterns as follows.
      \begin{itemize}
      	\item \textbf{Hiding the search pattern.} 
      	Given an encrypted skyline query request $\llbracket \mathbf{q}\rrbracket^A$, each cloud server $C_i,i\in\{1,2\}$ only receives the share $\langle \mathbf{q} \rangle^A_i$.
      	According to the security of additive secret sharing, it is ensured that encrypting the same query multiple times will produce different secret shares that are indistinguishable from uniformly random values. 
      	Therefore, given the security of additive secret sharing  \cite{demmler2015aby}, {\csa} cannot determine whether a new skyline query has been issued before.
      	Therefore, {\main} can hide the search pattern.
      	
      	\item \textbf{Hiding the access pattern.} 
      	The access pattern in fact indicates whether a tuple in the original database $\llbracket \mathbf{P}\rrbracket^A$ or the mapped database $\llbracket \mathbf{T}\rrbracket^A$  is a skyline tuple. That is, it refers to which tuples will appear in the query result $\llbracket\mathcal{SK}_{\mathbf{q}}\rrbracket^A$.
      	Since the skyline tuples are obliviously fetched in $\mathsf{secFetch}$ and obliviously filtered out in $\mathsf{secFilt}$, {\csa} cannot know which tuples are the skyline tuples.
      	 Therefore, {\main} can hide the access pattern.
      \end{itemize}

    \section{Experiments}
    \label{sec:expe}

    \subsection{Experiment Setup}
    We implement our protocol in C++. 
    All experiments are conducted on a machine with 8 AMD Ryzen 7 5800H CPU cores and 16 GB RAM running 64-bit Windows 10. 
    In our experiments, the two cloud servers are simulated by threads and executed in parallel on the same machine.
    The network delay is set to 1 ms.
    Similar to the previous works \cite{liu2019secure,ding2021efficient}, we use three synthetic datasets and a real-world NBA dataset. 
    Specifically, we generate independent (INDE), correlated (CORR), and anti-correlated (ANTI) datasets following \cite{liu2019secure}.
    We also build a real-world dataset about NBA players based on the data from the Kaggle dataset\footnote{\url{https://www.kaggle.com/drgilermo/nba-players-stats/data}.}, where each player has six attributes: minutes, points, rebounds, assists, blocks, and steals.
    The results reported in our experiments are the average over 100 skyline queries unless otherwise stated.

    \subsection{Evaluation on Accuracy}
    We first report the accuracy of {\main} to demonstrate the effectiveness of our design.
    Specifically, we first implement the plaintext dynamic skyline query algorithm (i.e., Algorithm \ref{alg:skyline_plaintext}).
    Then, over different datasets, we randomly generate 1000  skyline queries and use the plaintext algorithm and {\main} to search the skyline tuples with respect to these skyline queries.
    We use the skyline tuples output by the plaintext algorithm as the baseline to evaluate the accuracy of {\main}.
    If the skyline query results returned by {\main} exactly match that returned by the plaintext baseline, there is no accuracy loss so the accuracy is measured to be $100\%$. 
    The experiment results on different datasets are summarized in Table \ref{table:Accu}, where $n$ represents the number of tuples and $m$ represents the number of dimensions.
    It can be observed that {\main} outputs exactly the same skyline tuples as the plaintext algorithm, which validates the effectiveness of {\main}.

    \begin{table}[t!]
    \centering
        \caption{Accuracy of {\main} over Plaintext Baseline on Different Datasets}
        \label{table:Accu}
        \begin{tabular}{ccccc}  
            \toprule
            Dataset ($n$=1000, $m$=6)&CORR&INDE&ANTI&NBA\\
            \midrule
            Accuracy&100\%&100\%&100\%&100\%\\
            \bottomrule
        \end{tabular}
    \end{table}

    \subsection{Evaluation on Performance}

        \subsubsection{Evaluation on Query Latency}
        
         \begin{figure}[!t]
        	\centering
        	\begin{minipage}[t]{0.49\linewidth}
        		\centering{\includegraphics[width=\linewidth]{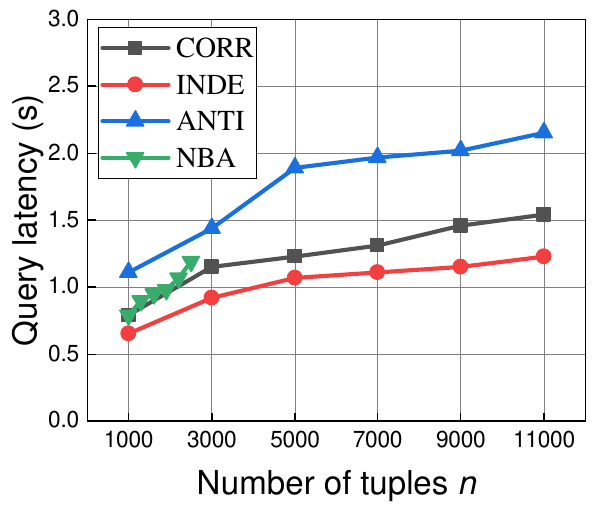}}
        		\caption{Query latency on different datasets, for varying number of tuples $n$ (with the number of dimensions $m=2$).}
        		\label{fig:ParaN}
        	\end{minipage}
        	\begin{minipage}[t]{0.49\linewidth}
        		\centering{\includegraphics[width=\linewidth]{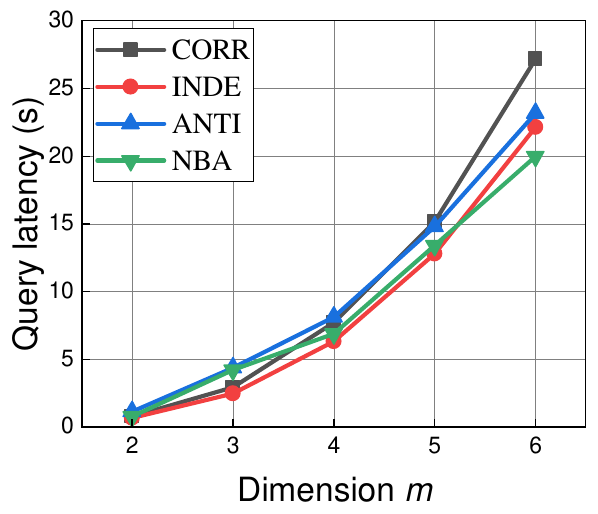}}
        		\caption{Query latency on different datasets, for varying number of dimensions $m$ (with the number of tuples $n = 1000$).}
        		\label{fig:ParaM}
        	\end{minipage}
        \end{figure}
        
        \label{sec:exp_Eval_Comp}
        We now examine the query latency of {\main}. That is, given an encrypted skyline query, we evaluate how long it takes {\csa} to obliviously execute dynamic skyline search on the encrypted database and output encrypted skyline tuples.
        We start with evaluating {\main} on different datasets with the number of dimensions $m=2$, for varying the number of tuples $n$, and summarize the results in Fig. \ref{fig:ParaN}.
        It is noted that since the NBA dataset only has 2500 tuples, the curve about its results is shorter than other datasets.
        From Fig. \ref{fig:ParaN}, it can be observed that under the same $m$, the query latency is small even with $n$ ranging from $1000$ to $11000$.
         % increases approximately linearly as $n$ increases.
        %
        \revise{Specifically, as $n$ increases from $1000$ to $11000$, the query latency on different datasets with $m=2$ increases from about $1$  to $2.15$ seconds.
        Moreover, it can be observed that the query latency on ANTI datasets is larger than others.
        The reason is that the dataset tuples in ANTI datasets show weaker correlation, and thus skyline queries on them produce more skyline tuples, which require more rounds of secure skyline tuples search.}
        Then, over different datasets and with the number of tuples $n=1000$, we evaluate {\main} for varying the number of dimensions $m\in\{2,3,4,5,6\}$, and summarize the results in Fig. \ref{fig:ParaM}.
        \revise{It can be observed that under the same $n$, as $m$ increases, the query latency grows quickly.
        Specifically, as $m$ increases from $2$ to $6$, the query latency on different datasets with $n=1000$ increases from about $1$ to $27$ seconds.
      	}

        \subsubsection{Evaluation on Communication Performance}

        We now examine the online communication performance of {\main}. That is, given an encrypted skyline query, the amount of data communicated between {\csa} to conduct the secure skyline query processing as per our design and output encrypted skyline tuples.
        Note that we use the same experiment setting as that in Section \ref{sec:exp_Eval_Comp}, and summarize the results in Fig. \ref{fig:Comm_vary_n} and Fig. \ref{fig:Comm_vary_m}.
        %
        % \revise{From the results, it can be observed that the bandwidth consumption on different datasets increases linearly with $n$ and $m$.
        %
        \revise{According to Fig. \ref{fig:Comm_vary_n}, as $n$ increases from $1000$ to $11000$, the communication cost over different datasets with $m=2$ increases from about $6$ to $144$ MB;.
        According to Fig. \ref{fig:Comm_vary_m}, as $m$ increases from $2$ to $6$, the communication cost over different datasets with $n=1000$ increases from about $6$ to $524$ MB.
        So the number of dimensions $m$ heavily affects the communication cost.
        }

  \begin{figure}[!t]
               	\centering
               	\begin{minipage}[t]{0.49\linewidth}
               		\centering{\includegraphics[width=\linewidth]{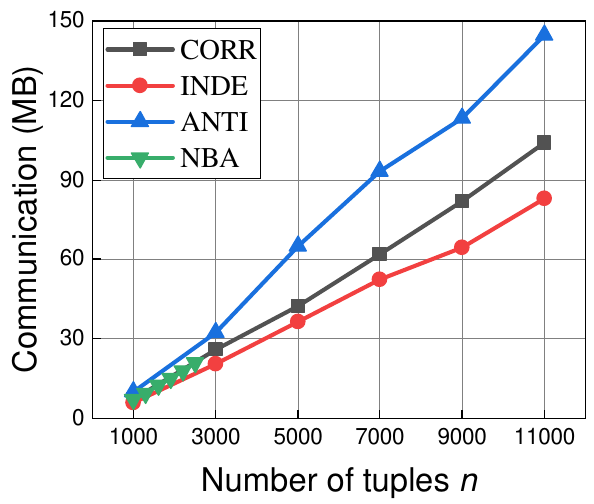}}
               		\caption{Communication cost for varying $n$, over different datasets (with $m=2$).}
               		\label{fig:Comm_vary_n}
               	\end{minipage}
               	\begin{minipage}[t]{0.49\linewidth}
               		\centering{\includegraphics[width=\linewidth]{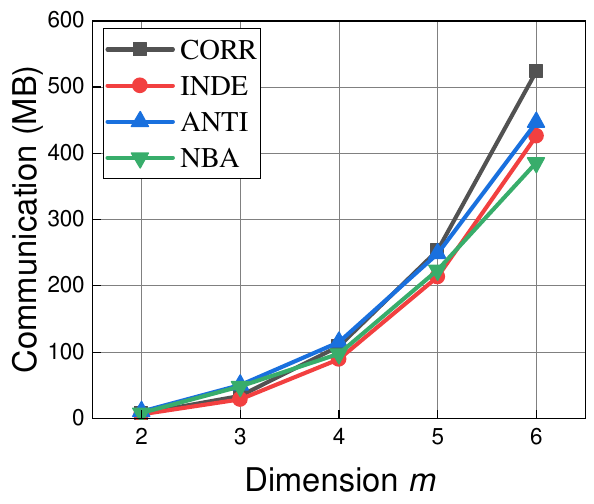}}
               		\caption{Communication cost for varying number of dimensions $m$, over different datasets (with $n = 1000$).}
               		\label{fig:Comm_vary_m}
               	\end{minipage}
               \end{figure}
  
    \subsection{Scalability Evaluation}

    To demonstrate the scalability of {\main} on large-scale datasets, we now report the computation cost of {\main} on larger datasets and different numbers of threads.
    Specifically, we first evaluate {\main} on different datasets under $m=2$ and single thread, for varying number of tuples $n\in\{2\times10^{5},3\times10^{5},4\times10^{5},5\times10^{5},6\times10^{5}\}$, and summarize the results in Fig. \ref{fig:LargeN}.
    It can be observed that even on $6\times10^{5}$ tuples, the query latency is still on the order of seconds, which should be tolerable for the client.
    We then evaluate {\main} on the CORR dataset under $m=2$, for varying number of tuples $n\in\{2\times10^{5},3\times10^{5},4\times10^{5},5\times10^{5},6\times10^{5}\}$ and the number of threads $\chi \in\{1,2,4\}$.
    The results are summarized in Fig. \ref{fig:Pall}.
    To reduce the time cost, we use the data partitioning method \cite{liu2019secure} to support parallel execution of our protocol.
    In the experiment, we divide the dataset into $\chi$ sub-datasets and distribute them to $\chi$ sub-threads.
    Each sub-thread runs {\main} protocol independently to compute the skyline tuples of the sub-dataset and sends the candidate set of skyline tuples to a main thread.
    After receiving the results from any two sub-threads, the main thread merges them into a new dataset and assigns it to an unoccupied
    sub-thread to securely search the skyline tuples.
    This process is repeated until the encrypted skyline tuples of the last merged
    dataset has been securely found, and these tuples are the ultimate encrypted skyline tuples returned to the client.

  \begin{figure}[!t]
    	\centering
    	\begin{minipage}[t]{0.49\linewidth}
    		\centering{\includegraphics[width=\linewidth]{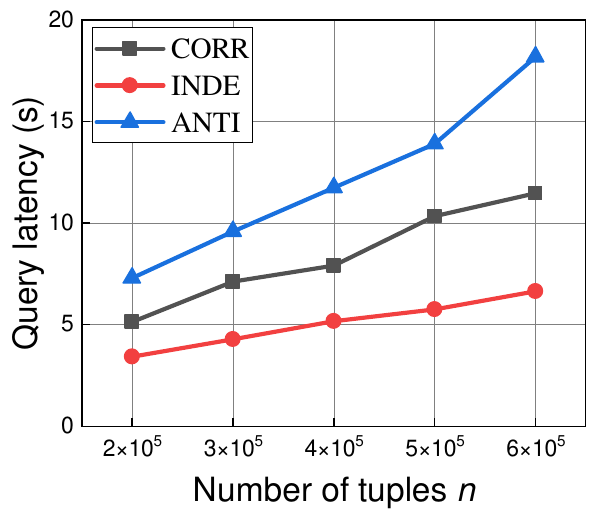}}
    		\caption{Query latency with $m=2$ and single thread, for varying large number of tuples $n$.}
    		\label{fig:LargeN}
    	\end{minipage}
    	\begin{minipage}[t]{0.49\linewidth}
    		\centering{\includegraphics[width=\linewidth]{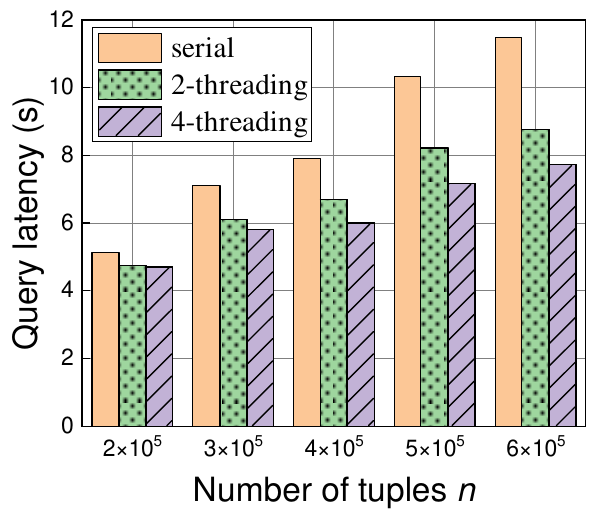}}
    		\caption{Query latency with $m=2$, for varying large number of tuples $n$ and number of threads $\chi$.}
    		\label{fig:Pall}
    	\end{minipage}
    \end{figure}

\subsection{Comparison to State-of-the-art Prior Works}

As reported in the state-of-the-art prior works, the protocols FSSP \cite{liu2019secure} and SMSQ \cite{ ding2021efficient} run secure skyline queries over small-scale datasets (e.g., CORR, INDE, and ANTI with size $n=1000, m=2$) in at least 1000 seconds and 100 seconds\footnote{Results are deduced from the figures in their papers since they do not give the exact values within text.}, respectively.
Meanwhile, we note that they consider the query latency to be the sum of computation time and memory copying time between threads, and do not consider network latency.
In contrast, even considering the network latency, {\main} only requires \textbf{0.79} seconds on CORR dataset, \textbf{0.65} seconds on INDE dataset, and \textbf{1.11} seconds on ANTI dataset, which is $\mathbf{90\sim 154\times}$ and $\mathbf{901}\sim\mathbf{1538\times}$ faster than SMSQ \cite{ding2021efficient} and FSSP \cite{liu2019secure}, respectively.
For the larger datasets (e.g., CORR, INDE, and ANTI with size $n=11000, m=2$), FSSP and SMSQ require at least 10000 seconds and 1000 seconds, respectively.
In contrast, even considering the network latency, {\main} only requires \textbf{1.54} seconds on CORR dataset, \textbf{1.23} seconds on INDE dataset, and \textbf{2.15} seconds on ANTI dataset, which is $\mathbf{465\sim 813\times}$ and $\mathbf{4651\sim 8130\times}$ better than SMSQ \cite{ding2021efficient} and FSSP \cite{liu2019secure}, respectively.

\section{Conclusion}
\label{sec:conclusion}
In this paper, we design, implement, and evaluate {\main}, a new system framework enabling fast privacy-preserving skyline query over outsourced encrypted cloud databases.
{\main} is fully based on the lightweight secret sharing technique, and is derived from a delicate synergy of three proposed secure components, including secure database mapping, secure skyline fetching, and secure skyline and dominated tuples filtering.
Extensive experiments over multiple datasets show that {\main} greatly improves upon state-of-the-art prior works \cite{liu2019secure,ding2021efficient} in query latency, with up to $8130\times$ improvement over FSSP \cite{liu2019secure} and up to $813\times$ improvement over SMSQ \cite{ding2021efficient}.

% allows the cloud to securely respond dynamic skyline queries within several seconds, improving upon the state-of-the-art schemes by $740\times$ and $7407\times$.
% %
% The experiment results demonstrate that {\main} is a more practical privacy-preserving dynamic skyline query system.
% %
% For future work, it would be interesting to explore how to leverage the recent advances in trusted hardware for further performance boost.

\section*{Acknowledgement}

This work was supported in part by the Guangdong Basic and Applied Basic Research Foundation under Grant 2021A1515110027, in part by the Shenzhen Science and Technology Program under Grants RCBS20210609103056041 and JCYJ20210324132406016, in part by the National Natural Science Foundation of China under Grant 61732022, in part by the Guangdong Provincial Key Laboratory of Novel Security Intelligence Technologies under Grant 2022B1212010005, in part by the Research Grants Council of Hong Kong under Grants CityU 11217819, 11217620, RFS2122-1S04, N\_CityU139/21, C2004-21GF, R1012-21, and R6021-20F, and in part by the Shenzhen Municipality Science and Technology Innovation Commission under Grant SGDX20201103093004019.

% This work was supported in part by the Guangdong Basic and Applied Basic Research Foundation (Grant No. 2021A1515110027), in part by the Shenzhen Science and Technology Program (Grants No. RCBS20210609103056041 and No. JCYJ20210324132406016), and in part by National Natural Science Foundation of China (Grant No.61732022). 

\bibliographystyle{IEEEtran}
\bibliography{my}

\begin{thebibliography}{10}
\providecommand{\url}[1]{#1}
\csname url@samestyle\endcsname
\providecommand{\newblock}{\relax}
\providecommand{\bibinfo}[2]{#2}
\providecommand{\BIBentrySTDinterwordspacing}{\spaceskip=0pt\relax}
\providecommand{\BIBentryALTinterwordstretchfactor}{4}
\providecommand{\BIBentryALTinterwordspacing}{\spaceskip=\fontdimen2\font plus
\BIBentryALTinterwordstretchfactor\fontdimen3\font minus
  \fontdimen4\font\relax}
\providecommand{\BIBforeignlanguage}[2]{{%
\expandafter\ifx\csname l@#1\endcsname\relax
\typeout{** WARNING: IEEEtran.bst: No hyphenation pattern has been}%
\typeout{** loaded for the language `#1'. Using the pattern for}%
\typeout{** the default language instead.}%
\else
\language=\csname l@#1\endcsname
\fi
#2}}
\providecommand{\BIBdecl}{\relax}
\BIBdecl

\bibitem{QinW0018}
Z.~Qin, J.~Weng, Y.~Cui, and K.~Ren, ``Privacy-preserving image processing in
  the cloud,'' \emph{{IEEE} Cloud Computing}, vol.~5, no.~2, pp. 48--57, 2018.

\bibitem{JiangWHWLSR21}
P.~Jiang, Q.~Wang, M.~Huang, C.~Wang, Q.~Li, C.~Shen, and K.~Ren, ``Building
  in-the-cloud network functions: Security and privacy challenges,''
  \emph{Proceedings of the IEEE}, vol. 109, no.~12, pp. 1888--1919, 2021.

\bibitem{Cox}
{Cox Automotive on AWS}, ``Cox automotive scales digital personalization using
  an identity graph powered by amazon neptune,''
  \url{https://aws.amazon.com/cn/blogs/database/cox-automotive-scales-digital-personalization-using-an-identity-graph-powered-by-amazon-neptune/},
  2020, [Online; Accessed 15-May-2022].

\bibitem{ADP}
{Automatic Data Processing on AWS}, ``Adp uses amazon neptune to model the
  agile modern workplace,''
  \url{https://aws.amazon.com/cn/solutions/case-studies/adp-neptune-case-study/?pg=ln&sec=c},
  2020, [Online; Accessed 15-May-2022].

\bibitem{airbnb}
{Airbnb on AWS}, ``Aws case study: Airbnb,''
  \url{{https://aws.amazon.com/solutions/case-studies/airbnb/?nc1=h_ls}}, 2018,
  [Online; Accessed 15-May-2022].

\bibitem{PIXNET}
{PIXNET on AWS}, ``Aws case study: {PIXNET},''
  \url{{https://aws.amazon.com/solutions/case-studies/pixnet/}}, 2014, [Online;
  Accessed 15-May-2022].

\bibitem{ghareh2018new}
J.~Ghareh~Chamani, D.~Papadopoulos, C.~Papamanthou, and R.~Jalili, ``New
  constructions for forward and backward private symmetric searchable
  encryption,'' in \emph{Proc. of ACM CCS}, 2018.

\bibitem{DautermanFLPS20}
E.~Dauterman, E.~Feng, E.~Luo, R.~A. Popa, and I.~Stoica, ``{DORY:} an
  encrypted search system with distributed trust,'' in \emph{Proc. of USENIX
  OSDI}, 2020.

\bibitem{sun2021practical}
S.-F. Sun, R.~Steinfeld, S.~Lai, X.~Yuan, A.~Sakzad, J.~K. Liu, S.~Nepal, and
  D.~Gu, ``Practical non-interactive searchable encryption with forward and
  backward privacy.'' in \emph{Proc. of NDSS}, 2021.

\bibitem{gui2023rethinking}
Z.~Gui, K.~G. Paterson, and S.~Patranabis, ``Rethinking searchable symmetric
  encryption,'' in \emph{Proc. of IEEE S\&P}, 2023.

\bibitem{mohamed2008skyline}
M.~E. Khalefa, M.~F. Mokbel, and J.~J. Levandoski, ``Skyline query processing
  for incomplete data,'' in \emph{Proc. of IEEE ICDE}, 2008.

\bibitem{balke2004efficient}
W.~Balke, U.~G{\"{u}}ntzer, and J.~X. Zheng, ``Efficient distributed skylining
  for web information systems,'' in \emph{Proc. of EDBT}, 2004.

\bibitem{huang2006skyline}
Z.~Huang, C.~S. Jensen, H.~Lu, and B.~C. Ooi, ``Skyline queries against mobile
  lightweight devices in {MANETs},'' in \emph{Proc. of IEEE ICDE}, 2006.

\bibitem{deng2007multi}
K.~Deng, X.~Zhou, and H.~T. Shen, ``Multi-source skyline query processing in
  road networks,'' in \emph{Proc. of IEEE ICDE}, 2007.

\bibitem{liu2017secure}
J.~Liu, J.~Yang, L.~Xiong, and J.~Pei, ``Secure skyline queries on cloud
  platform,'' in \emph{Proc. of IEEE ICDE}, 2017.

\bibitem{liu2019secure}
------, ``Secure and efficient skyline queries on encrypted data,''
  \emph{{IEEE} Transactions on Knowledge and Data Engineering}, vol.~31, no.~7,
  pp. 1397--1411, 2019.

\bibitem{ding2021efficient}
X.~Ding, Z.~Wang, P.~Zhou, K.-K.~R. Choo, and H.~Jin, ``Efficient and
  privacy-preserving multi-party skyline queries over encrypted data,''
  \emph{{IEEE} Transactions on Information Forensics and Security}, vol.~16,
  pp. 4589--4604, 2021.

\bibitem{bothe2014skyline}
S.~Bothe, A.~Cuzzocrea, P.~Karras, and A.~Vlachou, ``Skyline query processing
  over encrypted data: An attribute-order-preserving-free approach,'' in
  \emph{Proc. of International Workshop on Privacy and Secuirty of Big Data},
  2014.

\bibitem{wang2022efficient}
Z.~Wang, X.~Ding, H.~Jin, and P.~Zhou, ``Efficient secure and verifiable
  location-based skyline queries over encrypted data,'' \emph{Proc. VLDB
  Endow.}, vol.~15, no.~9, pp. 1822--1834, 2022.

\bibitem{wang2020stargazing}
J.~Wang, M.~Du, and S.~S.~M. Chow, ``Stargazing in the dark: Secure skyline
  queries with {SGX},'' in \emph{Proc. of DASFAA}, 2020.

\bibitem{zhang2022efficient}
S.~Zhang, S.~Ray, R.~Lu, Y.~Zheng, Y.~Guan, and J.~Shao, ``Towards efficient
  and privacy-preserving user-defined skyline query over single cloud,''
  \emph{{IEEE} Transactions on Dependable and Secure Computing}, 2022.

\bibitem{demmler2015aby}
D.~Demmler, T.~Schneider, and M.~Zohner, ``{ABY} - {A} framework for efficient
  mixed-protocol secure two-party computation,'' in \emph{Proc. of NDSS}, 2015.

\bibitem{borzsonyi2001skyline}
S.~B{\"o}rzs{\"o}nyi, D.~Kossmann, and K.~Stocker, ``The skyline operator,'' in
  \emph{Proc. of ICDE}, 2001.

\bibitem{kossmann2002shooting}
D.~Kossmann, F.~Ramsak, and S.~Rost, ``Shooting stars in the sky: An online
  algorithm for skyline queries,'' in \emph{Proc. of VLDB}, 2002.

\bibitem{papadias2005progressive}
D.~Papadias, Y.~Tao, G.~Fu, and B.~Seeger, ``Progressive skyline computation in
  database systems,'' \emph{{ACM} Transactions on Database Systems}, vol.~30,
  no.~1, pp. 41--82, 2005.

\bibitem{tao2006maintaining}
Y.~Tao and D.~Papadias, ``Maintaining sliding window skylines on data
  streams,'' \emph{{IEEE} Transactions on Knowledge and Data Engineering},
  vol.~18, no.~2, pp. 377--391, 2006.

\bibitem{liu2015finding}
J.~Liu, H.~Zhang, L.~Xiong, H.~Li, and J.~Luo, ``Finding probabilistic
  k-skyline sets on uncertain data,'' in \emph{Proc. of ACM CIKM}, 2015.

\bibitem{pei2007probabilistic}
J.~Pei, B.~Jiang, X.~Lin, and Y.~Yuan, ``Probabilistic skylines on uncertain
  data,'' in \emph{Proc. of VLDB}, 2007.

\bibitem{liu2015findinga}
J.~Liu, L.~Xiong, J.~Pei, J.~Luo, and H.~Zhang, ``Finding pareto optimal
  groups: Group-based skyline,'' \emph{Proceedings of the VLDB Endowment},
  vol.~8, no.~13, pp. 2086--2097, 2015.

\bibitem{yu2017fast}
W.~Yu, Z.~Qin, J.~Liu, L.~Xiong, X.~Chen, and H.~Zhang, ``Fast algorithms for
  pareto optimal group-based skyline,'' in \emph{Proc. ACM CIKM}, 2017.

\bibitem{hahnel2017high}
M.~H{\"a}hnel, W.~Cui, and M.~Peinado, ``High-resolution side channels for
  untrusted operating systems,'' in \emph{Proc. of USENIX ATC}, 2017.

\bibitem{van2017telling}
J.~Van~Bulck, N.~Weichbrodt, R.~Kapitza, F.~Piessens, and R.~Strackx, ``Telling
  your secrets without page faults: Stealthy page table-based attacks on
  enclaved execution,'' in \emph{Proc. of USENIX Security}, 2017.

\bibitem{lee2017inferring}
S.~Lee, M.-W. Shih, P.~Gera, T.~Kim, H.~Kim, and M.~Peinado, ``Inferring
  fine-grained control flow inside sgx enclaves with branch shadowing,'' in
  \emph{Proc. of USENIX Security}, 2017.

\bibitem{lee2020off}
D.~Lee, D.~Jung, I.~T. Fang, C.-C. Tsai, and R.~A. Popa, ``An off-chip attack
  on hardware enclaves via the memory bus,'' in \emph{Proc. of USENIX
  Security}, 2020.

\bibitem{DellisS07}
E.~Dellis and B.~Seeger, ``Efficient computation of reverse skyline queries,''
  in \emph{Proc. of VLDB}, 2007.

\bibitem{RiaziWTS0K18}
M.~S. Riazi, C.~Weinert, O.~Tkachenko, E.~M. Songhori, T.~Schneider, and
  F.~Koushanfar, ``Chameleon: {A} hybrid secure computation framework for
  machine learning applications,'' in \emph{Proc. of ACM AsiaCCS}, 2018.

\bibitem{mohassel2017secureml}
P.~Mohassel and Y.~Zhang, ``{SecureML}: A system for scalable
  privacy-preserving machine learning,'' in \emph{Proc. of IEEE S\&P}, 2017.

\bibitem{meng2018top}
X.~Meng, H.~Zhu, and G.~Kollios, ``Top-k query processing on encrypted
  databases with strong security guarantees,'' in \emph{Proc. of IEEE ICDE},
  2018.

\bibitem{chen2020metal}
W.~Chen and R.~A. Popa, ``Metal: {A} metadata-hiding file-sharing system,'' in
  \emph{Proc. of {NDSS}}, 2020.

\bibitem{zheng2022optimizing}
Y.~Zheng, C.~Wang, R.~Wang, H.~Duan, and S.~Nepal, ``Optimizing secure decision
  tree inference outsourcing,'' \emph{IEEE Transactions on Dependable and
  Secure Computing}, 2022, doi:\url{10.1109/TDSC.2022.3194048}.

\bibitem{du2020graphshield}
M.~Du, S.~Wu, Q.~Wang, D.~Chen, P.~Jiang, and A.~Mohaisen, ``Graphshield:
  Dynamic large graphs for secure queries with forward privacy,'' \emph{{IEEE}
  Transactions on Knowledge and Data Engineering}, 2020.

\bibitem{cui2020svknn}
N.~Cui, X.~Yang, B.~Wang, J.~Li, and G.~Wang, ``Svknn: Efficient secure and
  verifiable k-nearest neighbor query on the cloud platform,'' in \emph{Proc.
  of IEEE ICDE}, 2020.

\bibitem{wang2022privacy}
S.~Wang, Y.~Zheng, X.~Jia, and X.~Yi, ``Privacy-preserving analytics on
  decentralized social graphs: The case of eigendecomposition,'' \emph{IEEE
  Transactions on Knowledge and Data Engineering}, 2022, doi:
  \url{10.1109/TKDE.2022.3185079}.

\bibitem{zheng2020securely}
Y.~Zheng, H.~Duan, C.~Wang, R.~Wang, and S.~Nepal, ``Securely and efficiently
  outsourcing decision tree inference,'' \emph{IEEE Transactions on Dependable
  and Secure Computing}, vol.~19, no.~3, pp. 1841--1855, 2022.

\bibitem{wang2022PeGraph}
S.~Wang, Y.~Zheng, X.~Jia, and X.~Yi, ``{PeGraph}: A system for
  privacy-preserving and efﬁcient search over encrypted social graphs,''
  \emph{IEEE Transactions on Information Forensics and Security}, 2022, doi:
  \url{10.1109/TIFS.2022.3201392}.

\bibitem{Mozilla}
{Mozilla Security Blog}, ``{Next steps in privacy-preserving Telemetry with
  Prio.}'' online at \url{
  https://blog.mozilla.org/security/2019/06/06/next-steps-in-privacy-preserving-telemetry-with-prio/},
  2022.

\bibitem{WhitePaper}
{Apple and Google}, ``{Exposure Notification Privacy-preserving Analytics
  (ENPA) White Paper.}'' online at \url{
  https://covid19-static.cdn-apple.com/applications/covid19/current/static/contact-tracing/pdf/ENPA_White_Paper.pdf},
  2021, [Online; Accessed 1-Jun-2022].

\bibitem{WangWHZR16}
Q.~Wang, J.~Wang, S.~Hu, Q.~Zou, and K.~Ren, ``Sechog: Privacy-preserving
  outsourcing computation of histogram of oriented gradients in the cloud,'' in
  \emph{Proc. of ACM AsiaCCS}, 2016.

\bibitem{0002SKG19}
N.~Agrawal, A.~S. Shamsabadi, M.~J. Kusner, and A.~Gasc{\'{o}}n, ``{QUOTIENT:}
  two-party secure neural network training and prediction,'' in \emph{Proc. of
  ACM CCS}, 2019.

\bibitem{curtmola2006searchable}
R.~Curtmola, J.~A. Garay, S.~Kamara, and R.~Ostrovsky, ``Searchable symmetric
  encryption: improved definitions and efficient constructions,'' in
  \emph{Proc. of ACM CCS}, 2006.

\bibitem{liu2021medisc}
X.~Liu, Y.~Zheng, X.~Yuan, and X.~Yi, ``Medisc: Towards secure and lightweight
  deep learning as a medical diagnostic service,'' in \emph{Proc. of ESORICS},
  2021.

\bibitem{mohassel2018aby}
P.~Mohassel and P.~Rindal, ``{ABY}\({}^{\mbox{3}}\): {A} mixed protocol
  framework for machine learning,'' in \emph{Proc. of ACM CCS}, 2018.

\bibitem{cole1988parallel}
R.~Cole, ``Parallel merge sort,'' in \emph{Proc. of IEEE FOCS}, 1986.

\bibitem{lindell2017how}
Y.~Lindell, ``How to simulate it - a tutorial on the simulation proof
  technique,'' in \emph{Tutorials on the Foundations of Cryptography}, 2017,
  pp. 277--346.

\bibitem{canetti2000security}
R.~Canetti, ``Security and composition of multiparty cryptographic protocols,''
  \emph{Journal of Cryptology}, vol.~13, no.~1, pp. 143--202, 2000.

\bibitem{katz2005handling}
J.~Katz and Y.~Lindell, ``Handling expected polynomial-time strategies in
  simulation-based security proofs,'' in \emph{Proc. of TCC}, 2005.

\bibitem{curran2019procsa}
M.~Curran, X.~Liang, H.~Gupta, O.~Pandey, and S.~R. Das, ``Procsa: Protecting
  privacy in crowdsourced spectrum allocation,'' in \emph{Proc. of ESORICS},
  2019.

\bibitem{araki2016high}
T.~Araki, J.~Furukawa, Y.~Lindell, A.~Nof, and K.~Ohara, ``High-throughput
  semi-honest secure three-party computation with an honest majority,'' in
  \emph{Proc. of ACM CCS}, 2016.

\end{thebibliography}

%\balance
\end{document}